\documentclass[12pt]{article}

\usepackage{latexsym,amsmath,amssymb,amsthm,amsfonts,graphicx}
\newtheorem{example}{Example}
\usepackage{natbib}
\usepackage{epsfig,bm,authblk,url,blkarray}

\usepackage{float} 
\usepackage{booktabs} 
\usepackage{graphicx} 
\usepackage[margin=1cm]{caption} 
\usepackage{mathtools}
\usepackage{xcolor}
\usepackage{color}
\floatstyle{plain}
\usepackage{subfigure}
\usepackage{rotating}
\usepackage{makecell}

\usepackage{algcompatible}
\usepackage{algorithmicx}
\usepackage[noend]{algpseudocode}
\usepackage{algorithm}
\usepackage{diagbox}
\newfloat{Algorithm}{thp}{lop}
\floatname{Algorithm}{Algorithm}

\usepackage{titling}
\newcommand{\KL}[2]{\mathrm{KL}\!\left(#1 \,\|\, #2\right)}

\newtheorem{proposition}{Proposition}
\newtheorem{corollary}{Corollary}

\theoremstyle{remark}
\newtheorem{remark}{Remark}
\theoremstyle{plain}

\begin{document}

\title{Deep adaptive design with an evidential bias criterion}
\date{\empty}
\author{
  David Chen\thanks{\textit{Department of Statistics and Data Science, National University of Singapore}.}
  \and
  Michael Evans\thanks{\textit{Department of Statistical Sciences, University of Toronto, Toronto, ON M5S 3G3, Canada}.}
  \and
  Xinwei Li \thanks{\textit{Department of Civil and Environmental Engineering, National University of Singappore}.}
  \and
  Prateek Bansal\thanks{\textit{Department of Civil and Environmental Engineering, National University of Singapore} and \textit{Department of Statistics and Data Science, National University of SIngapore}.}
  \and
  David J. Nott\thanks{Corresponding author: \texttt{standj@nus.edu.sg}. \textit{Department of Statistics and Data Science, National University of Singapore} and \textit{Institute of Operations Research and Analytics, National University of Singapore}.}
}

\maketitle
\vspace{-0.7in}

\begin{abstract}
Bayesian optimal experimental design (BOED) aims to collect informative data by optimizing an expected utility reflecting the goals of an experiment.
However, this optimization is computationally challenging for common utilities and complex models.  This is especially so for sequential or adaptive designs, where design and data collection alternate, 
so that feedback from already observed data must be taken into account.  
Most existing BOED research employs information gain as the utility, leading to the expected information gain (EIG) criterion. While EIG is widely useful, it may not always adequately reflect experimental goals.
EIG can be viewed as rewarding experiments that produce large positive evidence for the truth on average, but it does not directly control the risk of an experiment producing misleading evidence.
Here we consider an alternative criterion, which we call bias against (BA), that prioritizes such 
control.  To address computational challenges when applying this criterion for
adaptive design, we consider a
policy-based deep adaptive design framework, which has previously been 
used for the EIG criterion.   
Minimizing a tractable upper bound on the BA objective is equivalent to
maximizing a variance-penalized EIG criterion, and we optimize the latter by
approximating it by Monte Carlo and learning design policies using
stochastic gradient methods. 
The differences between 
BA and EIG designs are demonstrated in several examples 
including the adaptive design of a complex discrete choice experiment.


\end{abstract}

\section{Introduction}\label{sec:Intro}

Collecting the most informative data possible is crucial to researchers in many different
fields, particularly when data collection is expensive.  Bayesian optimal experimental
design (BOED) proposes to choose the best design by maximizing an
expected utility.  Popular
design criteria are expectations, under the joint Bayesian model, of nonlinear
functions of the posterior density or marginal likelihood, quantities that are
intractable in most problems.  Because the intractabilities appear as an
inner expectation inside a nonlinearity, the criterion cannot be written as a
single expectation, and direct approximation requires nested Monte Carlo methods
\citep{rainforth+cyww18}.  This is referred to as the ``double intractability''
of the design objective.
Sequential or adaptive
design problems which alternate between observing data and choosing design variables
are particularly difficult, due to the need to 
incorporate the feedback from past data.
Recent progress in addressing the computational burden has used policy-based approaches, where a policy function specifying the design choice as a function of past data is optimized using techniques inspired by dynamic programming, reinforcement learning and deep learning \citep{huan+m16,foster+imr21,ivanova+fkgr21,shen+h23,blau+bcd22,blau+cdsb23,lim+nihg22}.    
These methods learn better designs than traditional greedy approaches and can be 
amortized, which means that after learning the policy the
computation of designs can be done very rapidly.  

Existing work on amortized adaptive design has focused mostly on using the expected information gain (EIG) criterion \citep{lindley56}.  While EIG is widely used, alternative criteria can be preferable in some situations, 
and here we explore an evidential bias criterion and the 
construction of designs
using the deep adaptive design approach of \cite{foster+imr21}.  
EIG can be thought of as rewarding large positive evidence
for the truth on average, whereas the evidential bias criterion considered here is a
probability of obtaining misleading evidence against the truth.  
The concern with controlling misleading evidence in 
design is not new. 
For example, \cite{royall97} argues that likelihood ratios are an appropriate way to quantify the strength of statistical evidence, and \citet{royall00} studies the probability that a likelihood ratio favours a false hypothesis, with
implications for design questions such as the choice of sample size.
\citet{blume02} extends the discussion to a range of study designs. 
The criterion we consider here addresses the same concern from a Bayesian 
perspective. It is built from the relative belief ratio, 
which measures the change in belief from prior to posterior, 
and the probability of misleading evidence is computed under the 
joint Bayesian model. 
Our concern in this work is with making such a criterion usable for 
amortized sequential design.
This paper makes four contributions.  First, we propose the use of evidential
bias -- also called the bias against (BA) in \cite{evans15} -- as a useful criterion in sequential
design problems.  Secondly, we develop an upper bound on the criterion that is 
insightful and easier to optimize than the original criterion using stochastic gradient methods.
Optimizing the upper bound is equivalent to maximizing a variance penalized version of the EIG
criterion, and is related to the risk-aware BOED approach of \cite{shen+h26}.  
Thirdly, we adapt the deep adaptive design framework of \cite{foster+imr21} to computation of designs
for the BA criterion.  Finally, we demonstrate how designs based on 
bias against differ from EIG-based designs in several examples, including the adaptive
design of a complex discrete choice experiment.  

Our contributions build on recent literature on computational methods for Bayesian sequential
design.  An important early work by \cite{huan+m16} computed non-myopic 
Bayesian adaptive designs by optimizing a 
policy for a finite-horizon Markov decision process.  They use an implicit representation of the optimal policy and solve for it
using approximate dynamic programming.   
\cite{foster+imr21} make the design policy explicit, 
and extend the contrastive bounds of \cite{foster+jotr20} to the sequential setting.  The bounds are easily
estimated by Monte Carlo and optimized using stochastic gradient approaches.  
Their deep adaptive design (DAD) approach
has been generalized in several directions, such as for implicit likelihoods \citep{ivanova+fkgr21,lim+nihg22}
and semi-amortized algorithms \citep{ivanova+hgr24}.  
\cite{blau+bcd22} combine a reinforcement learning formulation of Bayesian
adaptive design with the variational bounds in \cite{foster+imr21} and
use well-developed policy optimization methods from reinforcement learning to obtain
better designs, including for discrete design problems.  \cite{shen+h23} 
also consider reinforcement learning approaches, employing an
actor-critic policy gradient method for the optimization. 
\cite{shen+dh24} extend \cite{shen+h23} by employing
variational ``one point'' reward formulations and incorporating 
model uncertainty, nuisance parameters, handling of implicit models 
and more flexible design objectives.  
\cite{blau+cdsb23} address the difficulty of estimating the bounds of \cite{foster+imr21}
when the EIG is large, which they do by proposing a cross-entropy estimator 
based on a sequential version of the Barber-Agakov bound \citep{barber+a03}.  
A related approach is that of \cite{bracher+kiibr26}, where the authors avoid 
computation of the approximate posterior density and use diffusion approaches.  
\cite{phillips+kr26} obtain
singly intractable objectives using score-matching methods for optimizing 
the EIG criterion.    
%\cite{iollo+haf24} and \cite{iollo+haf25} develop sequential 
%design methods for the EIG which perform posterior sampling 
%alongside design optimization, the latter avoiding EIG lower bounds 
%altogether by sampling an expected posterior via diffusion-based 
%samplers within a bi-level optimization loop. Neither considers 
%policy-based methods. 

The above works focus mostly on EIG as the design criterion,
although some works such as \cite{shen+dh24} consider EIG for
subsets of parameters, model indicators or predictive
quantities of interest.  \cite{huang+wbka26} consider similar design
objectives which can be dynamically chosen, with amortization over
the choice of design and the posterior estimation.  
Many alternatives to EIG have been explored in BOED, but little of the 
existing work has been in the context of amortized sequential design.  
Some recent exceptions are \cite{huang+gak24}, who consider
design using predictive utilities which reflect downstream decision-making tasks, 
and \cite{rossa+pr26}, who formulate Bayesian design 
decision-theoretically in terms of expected future loss, leading to singly intractable objectives.  Their approach
requires two policies to be learnt, 
which they call the design and action policies.  

One goal of the present paper is to produce adaptive 
designs using an alternative
to EIG, which is focused on the avoidance of misleading evidence. 
The existing work most closely related to ours is 
\cite{shen+h26}, where the authors consider 
design criteria that penalize expected utility with
a utility variance.  They focus mostly on the prior-to-posterior Kullback-Leibler divergence as the utility, 
leading to a penalized EIG criterion.  
Penalizing an expected utility by its variance is one of 
many ways of expressing risk aversion in design, and
\citet[Section~6.2]{huan+jm24} discuss some of the alternatives, including
quantile-based measures such as the value-at-risk and 
conditional value-at-risk \citep{rockafellar+u00}, 
worst-case criteria and entropic risk.  They observe that these
approaches have so far seen relatively little use in nonlinear 
design with general utilities.
We consider penalized EIG criteria in our work, although we 
use a different variance
penalty term to \cite{shen+h26}, and focus on adaptive 
designs using policy learning.  
Our approach also has a different motivation
in terms of avoiding evidential bias, an interpretation that is helpful
in choosing the penalty parameter.
The connections between our approach and that of \cite{shen+h26} are discussed 
in detail in Section~\ref{sec:upper-bound}.  

In the next section we discuss the Bayesian formulation of optimal 
experimental design
in batch and sequential settings.  We then explain the EIG criterion and the 
alternative BA criterion that we focus on in our work.  
In Section~\ref{sec:upper-bound} we derive an upper bound on the BA criterion 
using the Paley-Zygmund inequality \citep{paley+z32}.  The bound has 
a simple expression in terms of the mean and variance of the log 
relative belief ratio, and we show that 
minimizing it is equivalent to maximizing a variance penalized EIG criterion, 
the latter being more convenient to optimize.  We also compare the bias
against, its upper bound and EIG theoretically.  
Section~\ref{sec:computation} discusses 
computational issues, while  
Section~\ref{sec:example} considers three examples, including the
adaptive design of a complex discrete choice experiment.  

\section{Bayesian optimal experimental design}\label{optimal}

First we briefly review basic ideas of BOED.   For further background, 
\cite{huan+jm24} give a comprehensive recent
review, and \cite{rainforth+fis24} and \cite{ryan+dmp16} discuss
computational aspects.  \cite{challoner+v95} 
consider classical design problems from the point of view of Bayesian decision theory.

\subsection{Batch and sequential design}

We introduce some notation first.  
Denote data to be collected as $y=(y_1,\dots, y_T)$ where $y_t\in \mathcal{Y}_t$, $t=1,\dots, T$.  There is an assumed parametric model for $y_t$ having parameters 
$\theta$, with a density 
$p(y_t|x_t,\theta)$, where $x_t\in \mathcal{X}_t$ is some design variable to be
chosen.  The prior density for $\theta$ is denoted $p(\theta)$.  
Write $x=(x_1,\dots, x_T)$, and the model for $y$ given $x,\theta$ is 
$$p(y|x,\theta)=\prod_{t=1}^T p(y_t|x_t,\theta).$$
We think of $t=1,\dots, T$ as indexing a sequence of experiments, and $x_t$ is chosen
before $y_t$ is observed.  The information available for choosing $x_t$ comes 
from past experiments, ${\cal I}_{t-1}=\{(y_j,x_j): j<t\}\in \mathcal{H}_t$ with ${\cal I}_0$ being the empty set.

In traditional ``batch'' experimental design we consider $T=1$ so that there
is a single experiment where the whole design $x$ is chosen before observing $y$.  
A utility function 
$U(y,x,\theta)$ is considered, and we choose $x$ by maximizing 
expected utility:  
$$x^*:=\arg \max_x E_{y\sim p(y|x)}\left(E_{\theta\sim p(\theta|y,x)}\left(U(y,x,\theta)\right)\right),$$
where
$$p(\theta|y,x)\propto p(\theta)\prod_{t=1}^T p(y_t|x_t,\theta),$$
is the posterior density for $\theta$, 
and $p(y|x)=E_{\theta\sim p(\theta)}(p(y|x,\theta))$.  
Equivalently, we can write
$$x^*=\arg \max_x E_{\theta\sim p(\theta)}\left(E_{y\sim p(y|x,\theta)}\left(U(y,x,\theta)\right)\right).$$
Throughout the manuscript, we will sometimes write the posterior density
$p(\theta|y,x)$ as $p(\theta|y)$ for simpler notation where no confusion
will arise in suppressing dependence on $x$.  

In contrast to batch experimental design, our work will focus on 
adaptive experimental design (also called sequential experimental design) 
where $T>1$ and at step $t$ the choice of $x_t$ can
depend on ${\cal I}_{t-1}$, so that feedback from previoius observations
must be accounted for in design decisions.    
Traditional (greedy) Bayesian adaptive design would make the choice of $x_t$ 
by maximizing the expected value of a utility 
$U_t(y_t,x_t,{\cal I}_{t-1},\theta)$:
$$x_t^*=\arg \max_{x_t} E_{y_t\sim p(y_t|{\cal I}_{t-1},x_t)}\left(E_{\theta\sim p(\theta|{\cal I}_{t})}\left(U_t(y_t,x_t,{\cal I}_{t-1},\theta)\right)\right),$$
where $p(y_t|{\cal I}_{t-1},x_t)=E_{\theta\sim p(\theta|{\cal I}_{t-1})}( p(y_t|\theta,x_t))$ and
$p(\theta|{\cal I}_{t})$ is the posterior density of $\theta$ given ${\cal I}_{t}$.  Equivalently, 
we can write
$$x_t^*=\arg \max_{x_t} E_{\theta\sim p(\theta|{\cal I}_{t-1})}\left(E_{y_t\sim p(y_t|x_t,\theta)}\left(U_t(y_t,x_t,{\cal I}_{t-1},\theta)\right)\right).$$
This traditional adaptive BOED approach has two main problems.  Firstly, 
computation is difficult because we need to repeatedly approximate the
posterior distribution after each experiment.  Secondly, because we aim
to only optimize the expected utility for the next experiment, we ignore
the effects of future data.  This leads to so-called ``greedy'' or ``myopic''
designs.  Amortized policy-based designs (e.g. \citealt{rainforth+fis24}) 
can address both of these issues.  

\subsection{Amortized sequential design}

Policy-based sequential design methods have
become quite popular recently, and these approaches
were first used for adaptive BOED 
by \cite{huan+m16} to the best of our knowledge.  In their
formulation the specification of the policy is implicit, and here
we will follow \cite{foster+imr21} who use an explicit 
policy function depending on learnable parameters $\varphi$, defining 
for any fixed
step $t$ a mapping $\eta_\varphi : \mathcal{H}_t \rightarrow \mathcal{X}_t$.
The construction of this mapping is important and should respect the symmetries of the
problem.  We will discuss this later.
\cite{foster+imr21} learn the parameters $\varphi$ to achieve the best global value of the
design criterion over the entire sequence of $T$ data-collection steps. This procedure does
not require approximating a posterior density at any stage, and, as noted above, it learns
better designs by optimizing the global design criterion 
directly rather than greedily optimizing
local criteria, which may yield inferior solutions for the global objective. A further
benefit is that, once the policy is learned, generating designs is very fast which is an important consideration in many applications.

\subsection{Expected information gain and relative belief}

\cite{foster+imr21} consider deep adaptive design for the expected information
gain criterion, denoted here as $\text{EIG}(x)$.  
We consider the batch (non-adaptive) case first.  
There are many equivalent formulations of 
EIG, and we begin with one expressed in terms of the prior-to-posterior 
Kullback-Leibler (KL) divergence \citep{kullback+l51}.  
For two distributions $Q$ and $\widetilde{Q}$ with densities $q$ 
and $\widetilde{q}$ respectively, the KL divergence of $Q$ 
from $\widetilde{Q}$ is 
$$\KL{\widetilde{q}(\theta)}{q(\theta)}
:=E_{\theta\sim \widetilde{q}(\theta)}(\log \widetilde{q}(\theta)-\log q(\theta)),$$
which measures how far $Q$ is from $\widetilde{Q}$ and is zero if and only if $\widetilde{Q}=Q$.  

The EIG can be defined as
\begin{align}
  \text{EIG}(x) 
  & := E_{y\sim p(y|x)}E_{\theta\sim p(\theta|y,x)}\left\{
  \log \frac{p(\theta|y,x)}{p(\theta)}\right\} \label{EIG2} \\
  & = E_{y\sim p(y|x)} \left\{ \KL{p(\theta|y,x)}{p(\theta)} \nonumber \right\},
\end{align}
so that $\text{EIG}(x)$ measures how far the prior is from 
the posterior on average for the design $x$.  If 
$\text{EIG}(x)$ is large this is good, since a large change in the
Bayesian update on average indicates that the experiment is highly informative.
We can immediately notice the potential
difficulty in optimizing $\mathrm{EIG}(x)$.  
In the KL divergence, evaluating $p(\theta|y,x)$ pointwise involves knowing
the normalizing constant $p(y|x)=E_{\theta\sim p(\theta)}(p(y|x,\theta))$, 
which is an intractable integral in the form of an expectation that appears
inside a log term.  Optimizing the nested expectation, which does not collapse
to a single joint expectation, is challenging.

An equivalent definition of $\mathrm{EIG}(x)$ that follows from \eqref{EIG2} is
\begin{align}
    \text{EIG}(x) & = E_{y\sim p(y|x)}(H(\theta)-H(\theta|y)),
\end{align}
where
$$H(\theta)=-E_{\theta\sim p(\theta)}(\log p(\theta))\;\;\mathrm{and} \;\; 
H(\theta|y)=-E_{\theta\sim p(\theta|y)}(\log p(\theta|y)),$$
which is the expected reduction in Shannon or differential entropy between
prior and posterior for the experiment $x$.  The quantity 
$H(\theta)-H(\theta|y)$ is called the information gain, and its use
as a utility for Bayesian experimental design was first suggested in 
\cite{lindley56}.  

The ratio that appears in \eqref{EIG2} is called the relative belief ratio \citep{evans15}, and we will introduce
some special notation for it, 
\begin{align} 
  \text{RB}(\theta;y,x):=\frac{p(\theta|y,x)}{p(\theta)}.  \label{relbelief}
\end{align}
The bias against criterion that is introduced in the next section is defined in terms of the relative belief ratio.  
The meaning of $\text{RB}(\theta;y,x)$ is the following:  if $\text{RB}(\theta;y,x)>1$ for a given $\theta$, then the data $y$ gives evidence in favour of $\theta$ being the true value.  
This is because $\text{RB}(\theta;y,x)>1$ implies that the posterior is larger than the prior at $\theta$, so
our belief that $\theta$ is the true value has increased.  Similarly, if $\text{RB}(\theta;y,x)<1$, there
is evidence against $\theta$ being the true value, whereas if $\text{RB}(\theta;y,x)=1$ then there
is no evidence either way.  We will call the set of parameters with evidence in favour, 
\begin{align}
  & \left\{\theta: \frac{p(\theta|y,x)}{p(\theta)}> 1\right\},  \label{plausible}
\end{align}
the plausible region.  
Introducing the notation
\begin{align}
  Z(y,x,\theta) & :=\log \text{RB}(\theta;y,x), \label{logrb}
\end{align}
we see from \eqref{EIG2} that $\text{EIG}(x)=E_{y\sim p(y|x), \theta\sim p(\theta|y,x)} (Z(y,x,\theta))$, 
the expected log relative belief.   

It is useful to have an alternative expression for the relative belief.  By rearranging Bayes' rule, 
\begin{align}
  \text{RB}(\theta;y,x) & =\frac{p(\theta|y,x)}{p(\theta)}=\frac{p(y|x,\theta)}{p(y|x)} \label{relbeliefalt}
\end{align}
and then using this alternative expression in \eqref{EIG2} we obtain a corresponding
alternative expression for $\text{EIG}(x)$, 
\begin{align}
  \text{EIG}(x) & = E_{\theta\sim p(\theta)}E_{y\sim p(y|\theta)}\left\{
  \log \frac{p(y|x,\theta)}{p(y|x)}\right\}. \label{EIG1} 
\end{align}

For design, we want to maximize $\text{EIG}(x)$.    \cite{foster+jotr20} consider
independent draws $\theta_0,\theta_1,\dots, \theta_L\sim p(\theta)$ and 
estimate $p(y|x)$ in \eqref{EIG1} by either
$$\frac{1}{L+1}\sum_{l=0}^L p(y|x,\theta_l)$$
or
$$\frac{1}{L}\sum_{l=1}^L p(y|x,\theta_l).$$
This results in two approximations to the EIG criterion, 
\begin{align}
  \text{EIG}_l(x) & = E_{\theta_0,\theta_1,\cdots, \theta_L\sim p(\theta)}E_{y\sim p(y|\theta_0)}\left\{
    \log \frac{p(y|x,\theta_0)}{\frac{1}{L+1}\sum_{l=0}^L p(y|x,\theta_l)}\right\}, \label{EIG-est1}
\end{align}
and
\begin{align}
  \text{EIG}_u(x) & = E_{\theta_0,\theta_1,\cdots, \theta_L\sim p(\theta)}E_{y\sim p(y|\theta_0)}\left\{
    \log \frac{p(y|x,\theta_0)}{\frac{1}{L}\sum_{l=1}^L p(y|x,\theta_l)}\right\}. \label{EIG-est2}
\end{align}
$\mathrm{EIG}_l(x)$ and $\mathrm{EIG}_u(x)$
are respectively lower and upper bounds for \eqref{EIG1}. The lower bound is the prior contrastive estimation (PCE) bound of \cite{foster+jotr20}, and the upper bound is the standard nested Monte Carlo estimator, whose upward bias follows from Jensen's inequality.   See \cite{rainforth+cyww18} for the convergence properties of such estimators. 
$\mathrm{EIG}_l(x)$ and $\mathrm{EIG}_u(x)$
are sometimes called contrastive bounds, because of the use of contrastive prior samples $\theta_1,\dots, \theta_L$.
\cite{foster+imr21} extend them to the sequential setting.

The bounds become tight as $L\rightarrow\infty$, and they can be optimized using stochastic
gradient ascent. 
For optimizing $\varphi$ with a policy mapping in the adaptive design context,
\cite{foster+imr21} use the lower bound
for reasons of numerical stability.  In this case, since 
$x_t$ is determined from $\varphi$ given ${\cal I}_{t-1}$ for all $t$, 
the design criteria \eqref{EIG-est1}
and \eqref{EIG-est2} are functions of $\varphi$.  To make this clearer, instead of 
writing $p(y|x,\theta)$ we write
$$p(y|\varphi,\theta)=\prod_{t=1}^T p(y_t|x_t=\eta_{\varphi}({\cal I}_{t-1}),\theta),$$
and 
$p(y|\varphi)=E_{\theta\sim p(\theta)} (p(y|\theta,\varphi))$. With this notation, 
we write \eqref{EIG-est1} as  
\begin{align}
  \text{EIG}_l(\varphi) & = E_{\theta,\theta_1,\cdots, \theta_L\sim p(\theta)}E_{y\sim p(y|\varphi,\theta)}\left\{
    \log \frac{p(y|\varphi,\theta)}{\frac{1}{L+1}\sum_{l=0}^L p(y|\varphi,\theta_l)}\right\}.\label{EIG-est12}
\end{align}
and it is this criterion that \cite{foster+imr21} optimize to learn the policy parameters $\varphi$.

\section{Bias against criterion and its upper bound}\label{sec:upper-bound}

As an alternative to EIG, we now consider the bias against criterion for design, the main object
of study for the rest of the paper. Again, we begin with the batch case.  
For a given parameter value $\theta$, the bias against $\theta$ is 
\begin{align}
  \text{BA}(\theta;x) & :=P_{y\sim p(y|x,\theta)}\left(\frac{p(y|\theta,x)}{p(y|x)}\leq 1 \right),
\end{align}
where $p(y|x)=E_{\theta\sim p(\theta)}(p(y|x,\theta))$.  
For a given $\theta$, $\text{BA}(\theta;x)$ is the probability that the true value of 
$\theta$ does not belong to the plausible region for data $y$ generated at $\theta$
using design $x$.
We write $\text{BA}(x)$ (which we call
simply the bias against) for its expectation, 
\begin{align} 
  \text{BA}(x) & := E_{\theta\sim p(\theta)}\left(\text{BA}(\theta;x)\right)=P_{\theta\sim p(\theta),y\sim p(y|x,\theta)}\left(\frac{p(y|\theta,x)}{p(y|x)}\leq 1\right). \label{bias-against}
\end{align} 
  $\text{BA}(x)$ is the average of $\text{BA}(\theta;x)$ over the prior for $\theta$, and
is the probability under the joint Bayesian model that misleading evidence is produced using design $x$.  
The discussion above generalizes to the case of a parameter of interest that is not the full
parameter $\theta$, but we don't consider this here for simplicity.   We can rewrite
\eqref{bias-against} as 
\begin{align}
  \text{BA}(x) & := E_{y\sim p(y|x)}\left\{P_{\theta\sim p(\theta|y,x)}\left(\frac{p(y|\theta,x)}{p(y|x)}\leq 1\right)\right\}.  \label{bias-against2}
\end{align}
which shows that minimizing $\text{BA}(x)$ corresponds to maximizing expected utility 
for the utility function 
$$-P_{\theta\sim p(\theta|y,x)}\left(\frac{p(y|\theta,x)}{p(y|x)}\leq 1\right).$$
Maximizing the expected value of this 
utility means having a small probability that the experiment using $x$ will produce
misleading evidence. 

Direct stochastic optimization of $\text{BA}(x)$ is difficult: it is a probability, 
for which the gradients are uninformative when the design is already reasonably 
good.  We therefore work with an upper bound on $\text{BA}(x)$ which is easier to 
optimize.  The bound is obtained from the 
Paley-Zygmund inequality \citep{paley+z32} and has an interpretation in 
terms of the mean and variance of the log relative belief ratio.  We first state 
the relevant inequality, then derive the bound, and finally describe computation 
in the policy-based setting.

\subsection{The Paley-Zygmund inequality and a coefficient-of-variation bound}
\label{sec:PZ}

Recall the notation $\text{RB}(\theta; y, x):= p(y|x,\theta)/p(y|x)$ 
for the relative belief ratio at $\theta$ when the 
data are $y$ and the design is $x$, and $Z(y,x,\theta):=\log 
\text{RB}(\theta; y, x)=\log p(y|x,\theta)-\log p(y|x)$ for its 
logarithm.  Throughout this section, expectations are taken under the joint Bayesian model $\theta\sim 
p(\theta)$, $y\sim p(y|x,\theta)$, and we write $E(\cdot)$ and 
$\mathrm{Var}(\cdot)$ for the joint expectation and variance.  With this 
notation, $\text{EIG}(x)=E(Z(y,x,\theta))$ and 
\begin{align*}
  \text{BA}(x) & = P(Z(y,x,\theta)\leq 0).
\end{align*}  

Our bound is derived from the following modified form of the Paley-Zygmund 
inequality.  It is usually stated for a non-negative random variable, 
but the proof only requires the mean to be non-negative.  
\begin{proposition}\label{prop:PZ}
Let $Z$ be a random variable with $0<\mathrm{Var}(Z)<\infty$ and $E(Z)\geq 0$.  Then for 
all $\delta\in[0,1]$,
\begin{align}
  P(Z>\delta E(Z)) & \geq \frac{(1-\delta)^2 E(Z)^2}{E((Z-\delta E(Z))^2)} 
  = \frac{(1-\delta)^2 E(Z)^2}{\mathrm{Var}(Z)+(1-\delta)^2E(Z)^2} 
  \label{eq:PZ-lower}\\
  P(Z\leq \delta E(Z)) & \leq \frac{E((Z-E(Z))^2)}{E((Z-\delta E(Z))^2)} = \frac{\mathrm{Var}(Z)}{\mathrm{Var}(Z)+(1-\delta)^2 E(Z)^2}. 
  \label{eq:PZ-upper}
\end{align}
\end{proposition}

The proof is an easy application of the 
Cauchy-Schwarz inequality, given in Appendix~\ref{app:PZ}.
Applying Proposition~\ref{prop:PZ} to $Z=Z(y,x,\theta)$ with $\delta=0$, and 
noting that $E(Z(y,x,\theta))=\text{EIG}(x)\geq 0$, we obtain the following.

\begin{corollary}\label{cor:EIG-bound}
If $0<\mathrm{Var}(Z(y,x,\theta))<\infty$, then 
\begin{align}
  \operatorname{BA}(x) & \leq \frac{\mathrm{Var}(Z(y,x,\theta))}{\mathrm{Var}(Z(y,x,\theta))+\operatorname{EIG}(x)^2}
  =: \widetilde{\operatorname{BA}}(x).
  \label{eq:BA-EIG-bound}
\end{align}
\end{corollary}

The quantity $\widetilde{\text{BA}}(x)$ defined in \eqref{eq:BA-EIG-bound} is the 
upper bound on $\text{BA}(x)$ that we use throughout the remainder of the paper.  
Rearranging \eqref{eq:BA-EIG-bound} gives  
\begin{align}
  \frac{\text{sd}(Z(y,x,\theta))}{\text{EIG}(x)} & \geq 
  \sqrt{\frac{\text{BA}(x)}{1-\text{BA}(x)}},
  \label{eq:CV-bound}
\end{align}
where $\text{sd}(Z(y,x,\theta))=\sqrt{\mathrm{Var}(Z(y,x,\theta))}$ is the standard
deviation of $Z(y,x,\theta)$, 
and the left-hand side of \eqref{eq:CV-bound} is the coefficient of 
variation of the log relative belief ratio under the joint Bayesian model.  This 
is a unitless measure of how reliably the experiment produces evidence in 
favour of the true parameter value.  The inequality \eqref{eq:CV-bound} implies
that if the coefficient of variation of the log relative belief ratio is small, 
then the bias against is also small.

\subsection{A variance-penalized EIG criterion}\label{sec:penalized}

To understand better what minimizing $\widetilde{\text{BA}}(x)$ does, divide 
numerator and denominator on the right-hand side of \eqref{eq:BA-EIG-bound} by 
$\mathrm{Var}(Z(y,x,\theta))$ (assuming this is non-zero) to obtain
\begin{align}
\widetilde{\text{BA}}(x)=\frac{1}{1+\text{EIG}(x)^2/\text{Var}(Z(y,x,\theta))}. \label{eq:BA-bound}
\end{align}
Since $\text{EIG}(x)=E(Z(y,x,\theta))\geq 0$, this shows that 
minimizing $\widetilde{\text{BA}}(x)$ is equivalent to maximizing 
\begin{align}
  R(x) & := \frac{\text{EIG}(x)}{\mathrm{sd}(Z(y,x,\theta))}
   = \frac{E(Z(y,x,\theta))}{\sqrt{\mathrm{Var}(Z(y,x,\theta))}},  
  \label{eq:R-criterion}
\end{align}
which is the inverse of the coefficient of variation appearing in \eqref{eq:CV-bound}.  
The numerator is the expected information gain, which rewards designs that 
produce large evidence for the true value on average.  The 
denominator $\mathrm{sd}(Z(y,x,\theta))$ penalizes variability of the log 
relative belief ratio, expressing the preference for reliable evidence.  
Maximizing EIG while ignoring its variability, which is what the standard 
EIG criterion does, may produce designs for which evidence is large on 
average but highly variable.  Optimizing for average behaviour, however, 
may not effectively control the risk of observing misleading evidence.
Maximizing $R(x)$ is equivalent to maximizing $\mathrm{EIG}(x)$ 
if the log relative belief ratio has a variance independent of the design, but this
is not the case in general.  

Maximizing $R(x)$ directly is awkward because estimating the ratio using Monte Carlo
can be unstable.  However, 
any maximizer $x^*$ of $R(x)$ maximizes $\text{EIG}(x)$ 
subject to a constraint on $\mathrm{Var}(Z(y,x,\theta))$.  A 
Lagrangian argument (see 
Appendix~\ref{app:Lagrangian}) shows that $x^*$ is also a maximizer of 
\begin{align}
  L_\alpha(x) & := E(Z(y,x,\theta)) - \alpha\, \mathrm{Var}(Z(y,x,\theta)), 
  \label{eq:L-alpha}
\end{align}
for some value $\alpha\geq 0$.  Using a penalized objective such as \eqref{eq:L-alpha} to 
optimize a ratio is a standard technique in fractional programming \citep{dinkelbach67}.  Dinkelbach's algorithm alternates between optimizing
$x$ for fixed $\alpha$ and optimizing $\alpha$ for fixed $x$.  
We do not do this here, since, as we explain below, the whole path of designs
traced out by varying $\alpha$ is of interest, not just the solution optimising
$R(x)$.
We therefore optimize $L_\alpha(x)$ over a grid 
of $\alpha$ values, and select the value from the grid maximizing the original ratio 
$R(x)$ only if we wish to choose a single $\alpha$.    

The penalized criterion 
\eqref{eq:L-alpha} is highly interpretable.  A similar criterion 
is considered by \cite{shen+h26}.  
At $\alpha=0$ \eqref{eq:L-alpha} reduces to EIG, and 
increasing $\alpha$ trades expected information gain (the expected log relative belief) 
for reduced variability of the log relative belief. 
The criterion of \cite{shen+h26} is not exactly the same as \eqref{eq:L-alpha}, 
as they consider a variance
penalty of the form (in our notation) $\mathrm{Var}_{y\sim p(y|x)}(\KL{p(\theta|y,x)}{p(\theta)})$, 
which is the variance of a prior to posterior KL divergence utility with respect to $p(y|x)$.  \cite{shen+h26} consider batch design, with \citet[Chapter 5]{shen23} discussing the adaptive case.  Further discussion of connections with their
approach is given in the next subsection.     

We can give an interesting interpretation of the entire 
path of solutions for the design
as $\alpha$ varies in \eqref{eq:L-alpha}, in terms of an extended notion of misleading evidence.    
Fix a constant $t\geq 0$ and consider 
the event $\{Z(y,x,\theta)\leq t\}$, that the log relative belief ratio fails 
to exceed $t$.  The choice $t=0$ recovers the bias against.  
Applying Proposition~\ref{prop:PZ} to $Z=Z(y,x,\theta)$ with 
$\delta=t/\text{EIG}(x)$, which lies in $[0,1]$ when 
$0\leq t\leq \text{EIG}(x)$, gives
\begin{align}
  P(Z(y,x,\theta)\leq t) & \leq \frac{\mathrm{Var}(Z(y,x,\theta))}{\mathrm{Var}(Z(y,x,\theta))+(\text{EIG}(x)-t)^2}.
  \label{eq:BA-threshold}
\end{align}
Minimizing the bound \eqref{eq:BA-threshold} is equivalent to maximizing 
\begin{align}
  R_t(x) & := \frac{\text{EIG}(x)-t}{\mathrm{sd}(Z(y,x,\theta))},
  \label{eq:R-t}
\end{align}
which differs from \eqref{eq:R-criterion} by replacing $\text{EIG}(x)$ with the 
offset $\text{EIG}(x)-t$.  
The optimizer of \eqref{eq:R-t} is also an optimizer of the penalized
criterion
$$L_{\alpha,t}(x)=E(Z(y,x,\theta))-t-
\alpha\, \mathrm{Var}(Z(y,x,\theta))$$
for some $\alpha$, and the family of designs traced out by optimizing
this criterion is the same as the family of solutions traced out by varying 
$\alpha$ in $L_{\alpha}(x)$.    
However if $\alpha$ in \eqref{eq:L-alpha} is chosen to tighten the bound 
\eqref{eq:BA-threshold} as much as possible, 
the design which is chosen varies with $t$.  So different choices of 
$\alpha$ correspond to optimizing a bound for $P(Z(y,x,\theta)\leq t)$
for different choices of $t$.

Proposition~\ref{prop:Lagrangian} in Appendix~\ref{app:Lagrangian} links 
maximizers of 
$R_t(x)$ to maximizers of $L_\alpha(x)$.  Although we assumed 
$\mathrm{EIG}(x)>t$ in application of the Paley-Zygmund inequality, 
it is not necessary to restrict the domain of optimization using
this condition.  Provided the design class
contains at least one $x$ with $\mathrm{EIG}(x)\ge t$, so that
$\max_x R_t(x)\ge0$, any design with $\mathrm{EIG}(x)<t$ has $R_t(x)<0$
and cannot be a maximizer. At an 
interior maximizer $x^*$ the corresponding penalty is
\begin{align}
  \alpha^*(t) & = \frac{\text{EIG}(x^*)-t}{2\,\mathrm{Var}(Z(y,x^*,\theta))},
  \label{eq:alpha-star}
\end{align}
obtained by applying Proposition~\ref{prop:Lagrangian} with offset $t$, where 
$x^*=x^*(t)$ is the maximizer of $R_t$.  
Writing 
$r^*(t):=\max_x R_t(x)$, and since $R_t(x)=(\text{EIG}(x)-t)/\sqrt{\mathrm{Var}(Z(y,x,\theta))}$ 
is decreasing in $t$ for fixed $x$, it follows that 
$r^*(t)$ is non-increasing in $t$.  
Choosing $\alpha$ in 
\eqref{eq:L-alpha} corresponds to choosing a 
threshold $t$ for the misleading-evidence event $\{Z(y,x,\theta)\leq t\}$ whose 
bound \eqref{eq:BA-threshold} is being tightened.  In practice 
\eqref{eq:alpha-star} cannot be set in advance, as it depends on the mean and
variance of $Z$ at the optimum.   For the computations in Section~\ref{sec:computation}, we 
search over a grid for $\alpha$,  
but the above discussion clarifies that each penalty 
value corresponds to tightening the bound on a particular threshold event.

\subsection{Relationship to risk-aware BOED}

The discussion of the previous subsection leads to connections between
our approach and the risk aware BOED method of \cite{shen+h26}.  
It is possible to justify their
criteria as approximate expected utility maximization using
an upper bound based on the Paley-Zygmund inequality.  To see this, suppose
that $U(y,x)$ is a utility function, with $E(U(y,x))\geq 0$ and
$0<\mathrm{Var}(U(y,x))<\infty$ with $y\sim p(y|x)$.  \cite{shen+h26} focus mostly 
on the choice $U(y,x)=\KL{p(\theta|y)}{p(\theta)}$ 
in their work, and this utility leads to the penalized EIG criterion.
Now consider obtaining a design by minimizing
$P(U(y,x)\leq t)$ for $y\sim p(y|x)$, corresponding to maximizing a utility of
$-P(U(y,x) \leq t)$.   For $0\leq t\leq E(U(y,x))$, applying the Paley-Zygmund inequality similar to before, gives an upper bound
\begin{align*}
  P(U(y,x)\leq t) \leq \frac{\mathrm{Var}(U(y,x))}
    {\mathrm{Var}(U(y,x))+(E(U(y,x))-t)^2},
\end{align*}
and minimizing this is equivalent to maximizing
$$\frac{E(U(y,x))-t}{\sqrt{\mathrm{Var}(U(y,x))}}.$$
An optimum of this upper bound is an optimum of the penalized criterion
$$E(U(y,x))-t-\alpha \mathrm{Var}(U(y,x)),$$
for some value of the penalty $\alpha>0$, 
where the family of designs traced out does not depend on $t$.  
In our bias against approach the choice of threshold $t=0$ for
tightening the bound is natural if we wish to choose 
a single $\alpha$.  This is because $t=0$ corresponds
to a log relative belief of $0$, or a relative belief of $1$, which has a special
meaning as the boundary between evidence in favour and evidence against.  
However, for the Kullback-Leibler utility and most other choices, there might be no
obvious choice of threshold $t$ to use if we choose $\alpha$ 
by tightening the bound. 

\subsection{Theoretical comparisons of $\mathrm{EIG}$ and $R$} \label{sec:eig-r-comparison}

Now that we have discussed the motivation for the $\mathrm{BA}$ criterion and its upper bound, we give some theoretical comparisons with EIG.  We start with a
simple example that is informative about the relationship between $\mathrm{EIG}(x)$, $\mathrm{BA}(x)$ and $R(x)$. In the example we take $x$ as fixed, although
the sample size could be taken as the design variable.  Since there is no 
varying $x$, we
write $\text{EIG}$, $\mathrm{BA}$ and $R$ 
for the EIG criterion, bias against and 
ratio criterion.  We have observed earlier that the 
BA upper bound $\widetilde{\mathrm{BA}}$ is a decreasing
function of $R$.  Consequently, by making $R$ large enough, 
the bias against can be made as close to zero as we like. 
No such guarantee holds for the EIG criterion, and our example demonstrates this.  
It shows a situation where $\mathrm{EIG}\rightarrow\infty$ but $\mathrm{BA}\rightarrow q>0$, so that making $\mathrm{EIG}$ large enough cannot 
by itself drive the bias against to zero.    
The lesson is that controlling the bias against requires controlling 
the variance of the log relative belief ratio in addition to its mean, the 
EIG, and this is what the ratio criterion does.

\begin{example}[EIG doesn't control bias against]
Consider a random sample $w_1,\dots, w_n\sim N(\theta,\sigma^2)$ where 
$\sigma^2>0$ is known and $\theta$ is unknown with a normal prior, 
$N(0,\tau^2)$.  Write $\bar{w}$ for the sample mean, with distribution
$\bar{w}|\theta\sim N(\theta,s^2)$, where $s^2:=\sigma^2/n$.  
Let $\xi$ be a Bernoulli
random variable with $Pr(\xi=1)=q$, where $q$ is known, and 
$B\sim N(0,\sigma_B^2)$ where $\sigma_B^2>0$, with $\xi$ and $B$ independent.  Suppose we observe
$y=\bar{w}$ if $\xi=0$ and $y=B$ if $\xi=1$.  Hence we observe a normal mean
with probability $1-q$, but with probability $q$ there is a failure, and 
in that case we observe
a random variable that is uninformative about $\theta$.  

Let us write $p_0(y|\theta)=\phi(y;\theta,s^2)$ for the density of $y$ given
$\theta,\xi=0$, and $p_1(y)=\phi(y;0,\sigma_B^2)$ for the density of 
$y$ given $\xi=1$.  
The likelihood is 
$$p(y|\theta)=(1-q)p_0(y\mid \theta)+q p_1(y),$$
and the prior predictive density of $y$, denoted $m(y)$, is
$$m(y)=(1-q)m_0(y)+qp_1(y),$$
where $m_0(y):=\phi(y;0,s^2+\tau^2)$.  
Observe that the function $\log\frac{(a+c)}{(b+c)}$ is monotonic in 
$c$, taking the value $\log(a/b)$ at $c=0$ and tending to zero as
$c\rightarrow\infty$.  We can write the log relative belief as 
$$Z=\log \frac{p(y|\theta)}{p(y)}=\log \frac{(1-q)p_0(y;\theta)+q p_1(y)}
{(1-q)m_0(y)+qp_1(y)},$$
which has the form of $\log\frac{(a+c)}{(b+c)}$ with $a=(1-q)p_0(y;\theta)$, $b=(1-q)m_0(y)$ and $c=qp_1(y)$.  
Hence $0\leq |Z|\leq |K|$ where $K=\log a/b=\log( p_0(y;\theta)/m_0(y))$, and
$Z$ and $K$ have the same sign.  This last fact means we can write the
bias against as 
$$\mathrm{BA}=P(Z\leq 0)=P(K\leq 0).$$
Since $K$ is the log relative belief for a normal location problem with normal
prior, we can compute it explicitly as
\begin{align}
  K & = \frac{1}{2}\log \left(1+\tau^2/s^2\right)-\frac{(y-\theta)^2}{2s^2}+\frac{y^2}{2(s^2+\tau^2)}.  \label{Kexpression}
\end{align}

Next, write
\begin{align}
 P(K\leq 0) & = P(\xi=0)P(K\leq 0|\xi=0)+P(\xi=1)P(K\leq 0|\xi=1) \nonumber \\
  & =(1-q)P(K\leq 0|\xi=0)+q P(K\leq 0|\xi=1).  \label{nloc-rb}
\end{align}
We show in Appendix \ref{app:example1} that
\begin{align}
 P(K\leq 0|\xi=1) & \geq 1-Cs(\log n)^{1/2}  \label{xi0condition}
\end{align}
where $C$ is a constant depending only on $\tau$ and $\sigma_B$,  
and
\begin{align}
  P(K\leq 0|\xi=0)\leq \sqrt{2}(s/\tau)^{1/2}.   \label{xi1condition}
\end{align}
Combining \eqref{nloc-rb}, \eqref{xi0condition} and \eqref{xi1condition}, 
and observing that $s=\sigma/\sqrt{n} \rightarrow 0$ as $n\rightarrow\infty$, 
we have the bias against satisfies
$$P(Z\leq 0)=P(K\leq 0)\rightarrow q,$$ 
as $n\rightarrow\infty$.  So the bias against cannot be made to go to zero
no matter how large $n$ is.  

In Appendix~\ref{app:example1} we also demonstrate that 
\begin{align}
    \frac{1-q}{2}\log\left(1+\frac{\tau^2}{s^2}\right)-H(q)\leq 
    \mathrm{EIG}\leq \frac{1-q}{2}\log\left(1+\frac{\tau^2}{s^2}\right), \label{EIG-inequality}
\end{align}
where $H(q)=-q\log q-(1-q)\log(1-q)$.  
Recalling that $s^2=\sigma^2/n$, we see that
$\mathrm{EIG}=(1-q)/2\log n + O(1)$.  This establishes 
that $\mathrm{EIG}\rightarrow\infty$. 

Finally let us deduce that both the EIG and its variance go to infinity with $n$, 
which is why EIG does not control bias against in this example and why criteria
such as $R$ which consider this variance are more suitable.  
From \eqref{eq:CV-bound}, and since we have shown that $\mathrm{BA}=q+o(1)$, 
we have 
\begin{align*}
 \frac{\mathrm{Var}(Z)}{\mathrm{EIG}^2} &\geq \frac{q}{1-q}+o(1),
\end{align*}
which gives 
\begin{align*}
    \mathrm{Var}(Z) & \geq \left(\frac{q}{1-q}+o(1)\right)\mathrm{EIG}^2.
\end{align*}
Hence $\mathrm{Var(Z)}\rightarrow\infty$ with $n$ since 
$\mathrm{EIG}\rightarrow\infty$ does, and   
\eqref{eq:CV-bound} also allows us to deduce that for the ratio
criterion $R$, $\lim \sup R^2\leq (1-q)/q$, and hence from
\eqref{eq:BA-bound}, $\lim \inf \widetilde{\mathrm{BA}} \geq q$, as should be 
the case because $\widetilde{\mathrm{BA}}$ is an upper bound on $\mathrm{BA}$.  
\end{example}

The example above is unusual in that the posterior fails to concentrate at
the true value with probability $q$.  
However, if a Bernstein-von Mises theorem holds, then we can say
something about the $\mathrm{EIG}$, $\mathrm{BA}$ and ratio 
criteria asymptotically.  
As in the example above, consider for the moment a fixed design
(suppressing $x$) with $y_1,\dots, y_n$ independent and identically distributed
from a distribution with parameter $\theta\in \mathbb{R}^d$ having prior density $p(\theta)$.  

Assume the conditions required for Theorem 2.1 of \cite{clarke+b90} to hold,
together with equality of the two information matrices appearing 
in the theorem, and we write the information matrix as  
$\mathcal{I}(\theta)$.
We also assume below it is valid to average limits over the prior.
Then the log relative belief ratio satisfies
\begin{align}
 Z & =\frac{d}{2}\log n+c(\theta)-\frac{W_n}{2}+o_p(1), \label{clarke-barron}
\end{align}
where $c(\theta):=\frac{1}{2}\log\det\mathcal{I}(\theta)-\frac{d}{2}\log(2\pi)
-\log p(\theta)$ and
$W_n:=S_n^{\top}\mathcal{I}(\theta)^{-1}S_n$, with
$S_n:=n^{-1/2}\sum_{i=1}^{n}\nabla_{\theta}\log p(y_i\mid\theta)$
the scaled score.   In this form, \eqref{clarke-barron} is the second limit of
Theorem~2.1 of \cite{clarke+b90} rearranged, and the $o_p(1)$ term
converges to zero in $L^1$ given $\theta$. Conditionally on almost
every $\theta$, $W_n\rightarrow\chi^2_d$ in distribution by the central
limit theorem for the score, and since the limit does not depend
on $\theta$, $W_n$ is asymptotically independent of $\theta$. 
Averaging over the prior gives the expansion
\begin{align*}
  \mathrm{EIG}=\frac{d}{2}\log\frac{n}{2\pi e}
  +\frac{1}{2}\,E\{\log\det\mathcal{I}(\theta)\}-E\{\log p(\theta)\}+o(1)
  \rightarrow\infty,
\end{align*}
an expansion of the same form as the one used by \citet{long+stw13} for fast
approximate computation of the EIG. 
In addition, 
$Z\rightarrow\infty$ in probability, so $\mathrm{BA}\rightarrow0$.  Under further regularity assumptions,
$\mathrm{Var}(Z)\rightarrow\mathrm{Var}\{c(\theta)\}+d/2<\infty$, so
$R\sim (\mathrm{Var}\{c(\theta)\}+d/2)^{-1/2}\,\mathrm{EIG}$.  This means that
we cannot have $\mathrm{EIG}\rightarrow\infty$ unless $R\rightarrow\infty$ as well.
The upper bound
$\widetilde{\mathrm{BA}}=(1+R^{2})^{-1}\rightarrow 0$, although
at the slow rate of $(\log n)^{-2}$.  

The discussion above does not directly apply to the policy-based designs
used in the rest of the paper, where the design policy
$\eta_\varphi$ is for a fixed number of experiments $T$ (Section~\ref{sec:policy-design}),
so that it isn't clear what $n$ is in the asymptotic results. It is natural
to treat the whole rollout of the $T$ experiments
given $\varphi$ as a single observation, and to have $n$ independent
replicates of this.  Write the independent trajectories
$y^{(1)},\dots,y^{(n)}$ as $y^{(j)}=(y_1^{(j)},\dots,y_T^{(j)})\sim
p(y\mid\varphi,\theta)$ given a common $\theta$. Since these are iid
given $\theta$, with $p(y\mid\varphi,\theta)$ playing the role of the
density for a single observation, the iid result above applies directly with
$n$ now the number of replicate experiments. Writing $\mathcal{I}(\varphi,\theta)$
for the information matrix of one replicate trajectory under $\varphi$,
and assuming the conditions of Theorem~2.1 of \cite{clarke+b90} hold 
together with equality of its two
information matrices, we obtain the expansion \eqref{clarke-barron} with
$\theta$, $\mathcal{I}(\theta)$ and $S_n$ reinterpreted as above and
$S_n:=n^{-1/2}\sum_{j=1}^n\nabla_\theta\log p(y^{(j)}\mid\varphi,\theta)$.
Here $c(\theta)$ becomes $c(\varphi,\theta)$, defined in the same way
with $\mathcal{I}(\varphi,\theta)$ in place of $\mathcal{I}(\theta)$.
Consequently $\mathrm{EIG}(\varphi)\to\infty$, $\mathrm{BA}(\varphi)\to0$
and $R(\varphi)\to\infty$ as $n\to\infty$ for fixed $\varphi$, with
$R(\varphi)=\{\tfrac{d}{2}\log n+a(\varphi)\}/s(\varphi)+o(1)$, where
$a(\varphi)$ is the policy-dependent constant in the $\mathrm{EIG}$
expansion and $s(\varphi)^2:=\mathrm{Var}\{c(\varphi,\theta)\}+d/2$.
Thus for large $n$ the EIG policy is governed by $a(\varphi)$ alone,
while the BA/ratio policy is governed by $s(\varphi)$, favouring policies
for which $\log\det\mathcal{I}(\varphi,\theta)$ is less variable over the
prior.  This asymptotic result is used here only to illustrate
why the ratio criterion is
sensitive to the variability of $\log\det\mathcal{I}(\varphi,\theta)$ over
the prior, unlike the EIG criterion.  Training our policies uses $n=1$, and in this case 
we need to study $\mathrm{Var}(Z(y,\varphi,\theta))$ directly. Another asymptotic
regime that could be studied would let $T\rightarrow\infty$, but this
is complicated because of the dependence induced by the adaptive design.

\subsection{Policy-based design with bias against}  \label{sec:policy-design}

Throughout the rest of the paper we work with the policy formulation for amortized
sequential design, where for a sequence of experiments indexed by $t=1,\dots, T$, we
make a design choice
$x_t$ at stage $t$ which 
is determined from ${\cal I}_{t-1}$ by a parametrized policy 
$\eta_\varphi$.  As in the EIG discussion above we write
$p(y|\varphi,\theta)=\prod_{t=1}^T p(y_t|x_t=\eta_\varphi({\cal 
I}_{t-1}),\theta)$ and $p(y|\varphi)=E_{\theta\sim p(\theta)}(p(y|\varphi,\theta))$, 
and we let $Z(y,\varphi,\theta):=\log p(y|\varphi,\theta)-\log p(y|\varphi)$.  
With this notation the optimization target we consider for optimizing
the upper bound $\widetilde{\mathrm{BA}}(x)$ is 
\begin{align}
  L_\alpha(\varphi) & = E(Z(y,\varphi,\theta)) - \alpha\, \mathrm{Var}(Z(y,\varphi,\theta)),
  \label{eq:L-alpha-phi}
\end{align}
and the ratio criterion is
\begin{align}
  R(\varphi) & = \frac{E(Z(y,\varphi,\theta))}{\sqrt{\mathrm{Var}(Z(y,\varphi,\theta))}}.  
  \label{eq:R-phi}
\end{align}
The bound \eqref{eq:BA-threshold}, the ratio criteria \eqref{eq:R-criterion} 
and \eqref{eq:R-t}, and the correspondence between ratio and penalized 
criteria all carry over to the policy formulation with $\varphi$ in place of 
$x$, since these results depend on the design only through the mean and 
variance of the log relative belief ratio.  See 
Remark~\ref{rem:policy-form} in Appendix~\ref{app:Lagrangian}.

\section{Computational issues}\label{sec:computation}

Both $L_\alpha(\varphi)$ and $R(\varphi)$ contain the intractable prior 
predictive density $p(y|\varphi)$ and hence must be approximated.  We first 
describe a Monte Carlo estimator of the objective, of the kind employed by 
\cite{foster+imr21} for EIG estimation, and then describe how it is used within 
a stochastic gradient optimization in which the Monte Carlo samples are 
refreshed at every step.

\subsection{Monte Carlo estimator of the objective}\label{sec:mc-objective}

Assume that a simulation from the joint Bayesian model can be expressed 
as a deterministic differentiable transformation of random 
variables $(\eta,\rho)$, written as 
$(\theta,y)=(G_1(\eta),G_2(\eta,\rho,\varphi))$.  
It is assumed that the distribution
of $(\eta,\rho)$ does not depend on $\varphi$.  
For example, they
could be standard normal variates.  
This kind of reparametrization 
is available for most of the 
priors and likelihoods used in our examples.

Draw a set of $J$ independent $(\rho,\eta)$ pairs,  
$(\eta^{(j)},\rho^{(j)})$, $j=1,\ldots,J$, and a further set of $L$ 
independent draws for $\eta$, $\widetilde\eta^{(l)}$, $l=1,\ldots,L$.  Define 
$\theta^{(j)}=G_1(\eta^{(j)})$, $y^{(j)}=G_2(\eta^{(j)},\rho^{(j)},\varphi)$, and 
$\widetilde\theta^{(l)}=G_1(\widetilde\eta^{(l)})$.  We estimate $p(y^{(j)}|\varphi)$ by 
\begin{align}
  \widehat p(y^{(j)}|\varphi) & := \frac{1}{L+1}\left\{ p(y^{(j)}|\varphi,\theta^{(j)})+
  \sum_{l=1}^L p(y^{(j)}|\varphi,\widetilde\theta^{(l)})\right\},
  \label{eq:p-hat}
\end{align}
and set 
\begin{align*}
  \widehat Z^{(j)}(\varphi) & := \log p(y^{(j)}|\varphi,\theta^{(j)})-\log\widehat p(y^{(j)}|\varphi).
\end{align*}
Writing 
\begin{align*}
  \overline Z(\varphi) & := \frac{1}{J}\sum_{j=1}^J \widehat Z^{(j)}(\varphi), \qquad
  S(\varphi)  := \frac{1}{J-1}\sum_{j=1}^J \left(\widehat Z^{(j)}(\varphi)-\overline Z(\varphi)\right)^2,
\end{align*}
our Monte Carlo objective for fixed $\alpha\geq 0$ is 
\begin{align}
  \widehat L_\alpha(\varphi) & = \overline Z(\varphi) - \alpha\, S(\varphi).
  \label{eq:L-alpha-hat}
\end{align}
with the corresponding estimate of the ratio criterion being 
$\widehat R(\varphi):=\overline Z(\varphi)/\sqrt{S(\varphi)}$.  The 
reparametrization makes $\widehat L_\alpha(\varphi)$ and $\widehat R(\varphi)$ 
differentiable in $\varphi$ for a fixed realization of the noise variables 
$\eta^{(j)},\rho^{(j)},\widetilde\eta^{(l)}$, so that gradients can be obtained by 
automatic differentiation and the policy can be efficiently 
trained by stochastic gradient ascent.  The mean term 
$\overline Z(\varphi)$ is a direct Monte Carlo estimate of 
the lower-bound estimator \eqref{EIG-est12}, so that with $\alpha=0$ 
\eqref{eq:L-alpha-hat} reduces to an estimate of the PCE bound 
similar to the one used in \cite{foster+imr21}.  For $\alpha>0$, the additional term penalizes the Monte Carlo variance of the estimated log relative belief ratio.

\subsection{Stochastic optimization with refreshed Monte Carlo samples}
\label{sec:refreshed}

At each stochastic-gradient update we draw a fresh
Monte Carlo batch, form the estimator \eqref{eq:L-alpha-hat}, take one gradient step, and then discard the batch.  Thus the optimization uses refreshed Monte Carlo approximations of the population objective, rather than a single fixed randomized approximation.

More precisely, with current iterate $\varphi^{(s)}$ and step size $\gamma_s$,
we update
\begin{align}
  \varphi^{(s+1)}
  &=
  \varphi^{(s)}
  +
  \gamma_s
  \nabla_\varphi
  \widehat L_\alpha(\varphi)
  \big|_{\varphi=\varphi^{(s)}},
  \label{eq:sgd-step}
\end{align}
where the Monte Carlo variables entering $\widehat L_\alpha$ are drawn
independently at each step $s$.  In our implementation the outer batch of size $J$ can be accumulated over several smaller minibatches, which reduces memory requirements.  The gradients of the minibatch contributions are accumulated before the parameter update.  When more than one minibatch is used, the variance
term $S(\varphi)$ is computed within each minibatch and averaged across
minibatches.  We use the unbiased form of the sample variance estimator 
in all cases, so this does not change the quantity being estimated, which
will be the variance of the log relative belief ratio when $L\rightarrow\infty$.  
Using a single minibatch of size $J$ recovers the expression 
given above.  

The gradient of $\widehat L_\alpha(\varphi)$ is not an unbiased estimator of 
$\nabla_\varphi L_\alpha(\varphi)$ for finite $J$ and $L$.  The optimization 
\eqref{eq:sgd-step} is therefore a stochastic gradient scheme with biased 
gradients, with the bias vanishing as $J,L\to\infty$.  Despite this bias, we 
find that small values of $J$ and $L$ are sufficient to learn good policies, 
while keeping computational and memory requirements low. 
A similar approximation, in which a biased 
but inexpensive gradient is used based on a batch approximation 
in place of a more costly unbiased one, was 
used successfully by \cite{chen+no25} in a variational inference context.  We 
use the Adam optimizer \citep{kingma+b15} for the updates \eqref{eq:sgd-step}. In some examples we also use periodic model
checkpointing for numerical stability.  The example-specific training details
are given in Appendix~\ref{app: procedure hyper}, where we also explore
sensitivity of the optimization to the tuning parameters in the stochastic
gradient algorithm.

\subsection{Choice of the penalty parameter}\label{sec:grid}

The discussion in Section~\ref{sec:penalized} shows that the penalty parameter
$\alpha$ has an evidential interpretation.  Varying $\alpha$ traces out designs
with different degrees of conservativeness in controlling misleading evidence.
Equivalently, different values of $\alpha$ can be viewed as corresponding to
different threshold events of the form $\{Z\leq t\}$ whose bounds are being
tightened.  In the examples below we consider a grid of possible values for
$\alpha$, 
\[
  \alpha \in \{0.5,0.75,1.0,1.25,1.5\}.
\]
For each value of $\alpha$, we train a policy by maximizing the Monte Carlo
objective \eqref{eq:L-alpha-hat}.  The resulting policies are not interpreted as
ordinary hyperparameter perturbations of a single objective.  Rather, they form
a path of designs from less conservative to more conservative policies, with
larger $\alpha$ placing greater emphasis on reducing the variability of the log
relative belief ratio. Finally, if we wish to choose a single $\alpha$, we do so by selecting the $\alpha$ attaining the highest estimated ratio $\widehat R(\varphi):=\overline Z(\varphi)/\sqrt{S(\varphi)}$ on average, evaluated on a random batch sample.  This corresponds to the tightest bound on the bias against.

For the source-location and hyperbolic-discounting examples of Sections 5.1 and 5.2 respectively, we use the
policy-network architecture of \cite{foster+imr21}.  The network embeds each previous design-outcome pair and aggregates these embeddings, so
that the proposed design depends on the accumulated experimental history rather than on an arbitrary ordering of the history elements.  For the pizza choice example of Section 5.3, we use the same history-aggregation structure as in the hyperbolic-discounting example, with input and output dimensions and hyperparameters adjusted to the design space.

\subsection{Gumbel-Softmax relaxation}
\label{sec:discrete-relaxation}

Our discussion so far assumes that the 
outcome $y$ can be written as a differentiable function $y=G_2(\eta,\rho,\varphi)$ 
with $\eta,\rho$ random variables having 
distributions not depending on $\varphi$, so that 
$\widehat L_\alpha(\varphi)$ depends smoothly on $\varphi$ for fixed 
draws of $(\eta,\rho)$.  When $y$ is discrete this does not hold 
(as in the hyperbolic discounting example of 
Section~5.2 below, where the response is Bernoulli). For these examples
we use differentiable relaxations during training
\citep{jang2016categorical,maddison2016concrete}.  Bernoulli or categorical
observations are replaced by Gumbel-Softmax variables, so that the
simulated outcome is a differentiable function of the policy parameter and
random noise.  The corresponding approximation is used in the Monte Carlo
objective during training. A description of the Gumbel-softmax trick is given in Appendix \ref{app:gumbel}. 

The Gumbel-softmax relaxation is used only for optimization and policy net 
training.  For the final design-trace
evaluations, we use the original discrete simulator. 
Thus the learned policy is trained using a smooth approximation 
that enables gradient-based optimization, while the reported 
adaptive design traces are generated under the intended discrete 
data-generating model. The precise simulator, likelihood
and any example-specific training stabilization are 
described in the individual examples.

\section{Examples}\label{sec:example}

We consider three examples. The first two are based on examples in \cite{foster+imr21}, while the third is a discrete choice experiment on a pizza choice. In each example we compare a policy trained for EIG with policies trained using the BA objective \eqref{eq:L-alpha-hat}. For the location-finding example the BA policies are trained from random initialization.
For the hyperbolic-discounting and pizza choice examples, we
use a common EIG warm start.  Starting from this same warm start, we then continue fine-tuning for additional epochs under either the EIG objective or the BA objective, resulting in the EIG and BA policies respectively that
are compared in the Figures. 

For visualization of adaptive design traces, we use the same general procedure
in all three examples.  For a fixed parameter value $\theta$, each policy is run
repeatedly, producing a sample of stochastic design traces.  We treat each trace
as an ordered sequence of designs.  If
$X^{(m)}=(x_1^{(m)},\ldots,x_T^{(m)})$ denotes the $m$th trace, we stack it into
a vector $v^{(m)}=\operatorname{vec}(X^{(m)})$ and define the representative
trace by the geometric median
\[
\widehat v
=
\arg\min_v
\sum_m \|v-v^{(m)}\|_2 .
\]
The vector $\widehat v$ is then reshaped to obtain the representative design
trace.  We compute this geometric median using the Weiszfeld iteration algorithm
\citep{weiszfeld37}.  The use of representative traces over many simulations
for the same $\theta$ is used to remove noise in order to make the 
differences between the design policies easier to recognize and understand.  

To compare policies trained by different objectives, we compute the dynamic time warping (DTW) distance between their representative traces \citep{JSSv031i07}.  We then sample a set of parameter values from the prior and display representative traces for parameter values corresponding to selected $(0\%,20\%,40\%,60\%,80\%,100\%)$ quantiles of the DTW distance, together with several manually chosen parameter values that illustrate different inferential regimes. The detailed hyperparameters for training and visualization are presented in Appendix~\ref{app: procedure hyper}.  Hereafter we call the policy trained by optimizing objective \eqref{EIG-est12} the EIG policy, and the policy trained by optimizing objective \eqref{eq:L-alpha-hat} the BA policy.  For the pizza choice example in Section \ref{subsec:discrete-choice}, the design at each step contains
two-dimensional attributes of three pizza options.  We compute the representative trace separately for each option here. We sum the three DTW distances for the options to get the overall DTW distance for selecting parameter values based on the quantiles. All codes and results are saved in \url{https://github.com/DZCQs/DAD-bias_against}.

\subsection{Source location finding}\label{subsec:source-location-finding}

We consider first 
the source location finding example given in \cite{foster+imr21}, based on a 
related model in \cite{sheng+h05}.  There are two sources in two-dimensional space, which
send signals where the intensity decays with distance according to an inverse squared law.  
The locations of the first and second source will be denoted by $\Theta_1=(\theta_1,\theta_2)^\top$
and $\Theta_2=(\theta_3,\theta_4)^\top$, with the unknown parameter to be inferred
$\theta=(\Theta_1^\top,\Theta_2^\top)^\top$.  A multivariate standard normal prior is assumed for $\theta$.  We can take $T$ measurements at different times
from a sensor that can be moved at each time.  The sensor location is a design variable, so we have to choose the location $x_t=(x_{t1},x_{t2})^\top$ 
adaptively at each time given
the previous locations and measurements.  At time $t$ we observe $y_t$, a noisy measurement
informative about the total log-intensity of the signal from the two sources.  Following
\cite{foster+imr21}, the total intensity at sensor location $x$ is
$$\mu(\theta,x)=0.1+\sum_{k=1}^2 \frac{1}{10^{-4}+\|\Theta_k-x\|_2^{2}},$$
and we model the measurement on the log-intensity scale with additive heavy-tailed noise,
$$y_t = \log\mu(\theta,x_t) + \sigma\,\varepsilon_t, \qquad \varepsilon_t\sim t_{\nu},$$
where $\varepsilon_t$ has a Student-$t$ distribution with $\nu=3$ degrees of freedom and
$\sigma=0.7$.  The heavy-tailed observation noise differs from the Gaussian noise used by
\cite{foster+imr21} and makes the inference problem more susceptible to occasional
misleading observations, which is the regime of interest for the bias against criterion.

\begin{figure}[!htbp]
    \centering
    \includegraphics[width=0.5\linewidth]{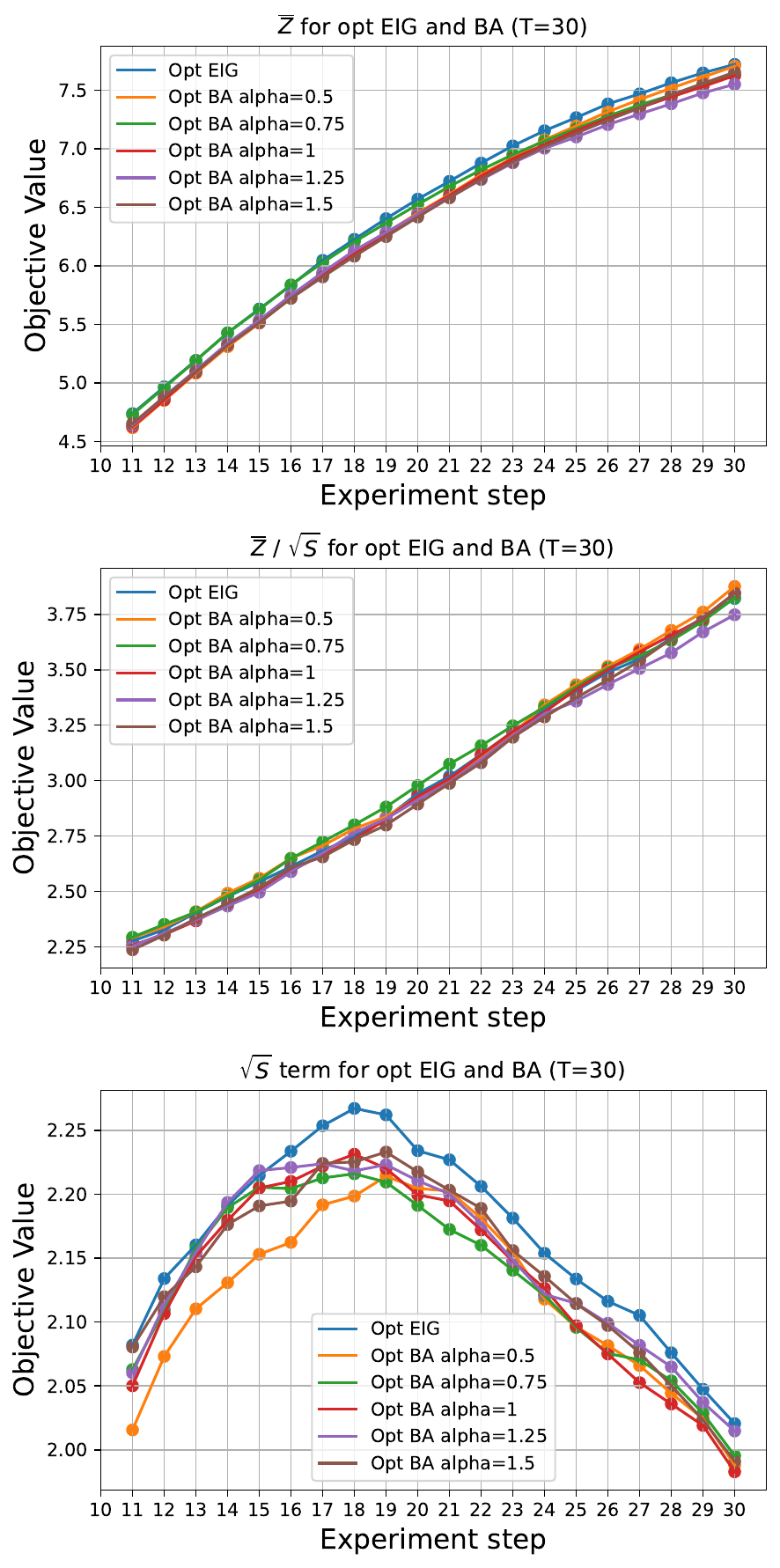}
    \caption{\small Evaluated $\overline{Z}$ (top), $\overline{Z}/\sqrt{S}$ (middle) and $\sqrt{S}$ (bottom) over $T=30$ experimental steps for location finding example for EIG policy and BA policies with $\alpha=0.5,0.75,1.0,1.25,1.5$.}
\label{fig loss: loc}
\end{figure}

Figure \ref{fig loss: loc} compares the EIG policy with BA policies. For each trained policy, we run a Monte Carlo evaluation batch and compute the sequential estimates $\overline Z(\varphi)$, $S(\varphi)$, and
$\widehat R(\varphi)=\overline Z(\varphi)/\sqrt{S(\varphi)}$ after each experimental step. The EIG policy gives the largest value of $\overline Z(\varphi)$, as expected, since it is trained directly to maximize the EIG lower bound.  The BA policies sacrifice some value of $\overline Z(\varphi)$, but reduce the estimated variability $S(\varphi)$ of the log relative belief ratio.  As a result, for suitable values of $\alpha$ the BA policy achieves a larger ratio
$\widehat R(\varphi)$ over most experimental steps. This is consistent with the comparison in
Section~\ref{sec:eig-r-comparison}: in this example the BA objective changes
the policy mainly through the design-induced variability of $Z$, reducing
$S(\varphi)$ while accepting some reduction in $\overline Z(\varphi)$.
  
Figure~\ref{fig design trace random: loc} plots representative
traces for different parameter values chosen from 
simulations from the prior based on quantiles of DTW distance 
as discussed above.  
For the BA policy we use $\alpha=0.75$, selected
according to the procedure described in Section~\ref{sec:grid}.
The selection of the parameter values by quantiles of DTW distance ensures 
that the cases range from those where the policies behave most similarly,
to those where they are most different.  
Both policies adapt to the underlying source locations and tend to take smaller
steps when the sensor is close to the signal sources.  
Although both policies show similar behaviour, in 5 of the 6 plots the
BA policy designs explore the design space more widely.  
Figure~\ref{fig design trace manual: loc} compares the policies for
two manually chosen parameter configurations, one where the two sources
are far apart, and one where they are close together.  When the sources
are close together, the EIG policy is more dispersed, 
whereas for the BA policy the opposite behaviour is observed.

\begin{figure}[!htbp]
    \centering
    \includegraphics[width=0.9\linewidth]{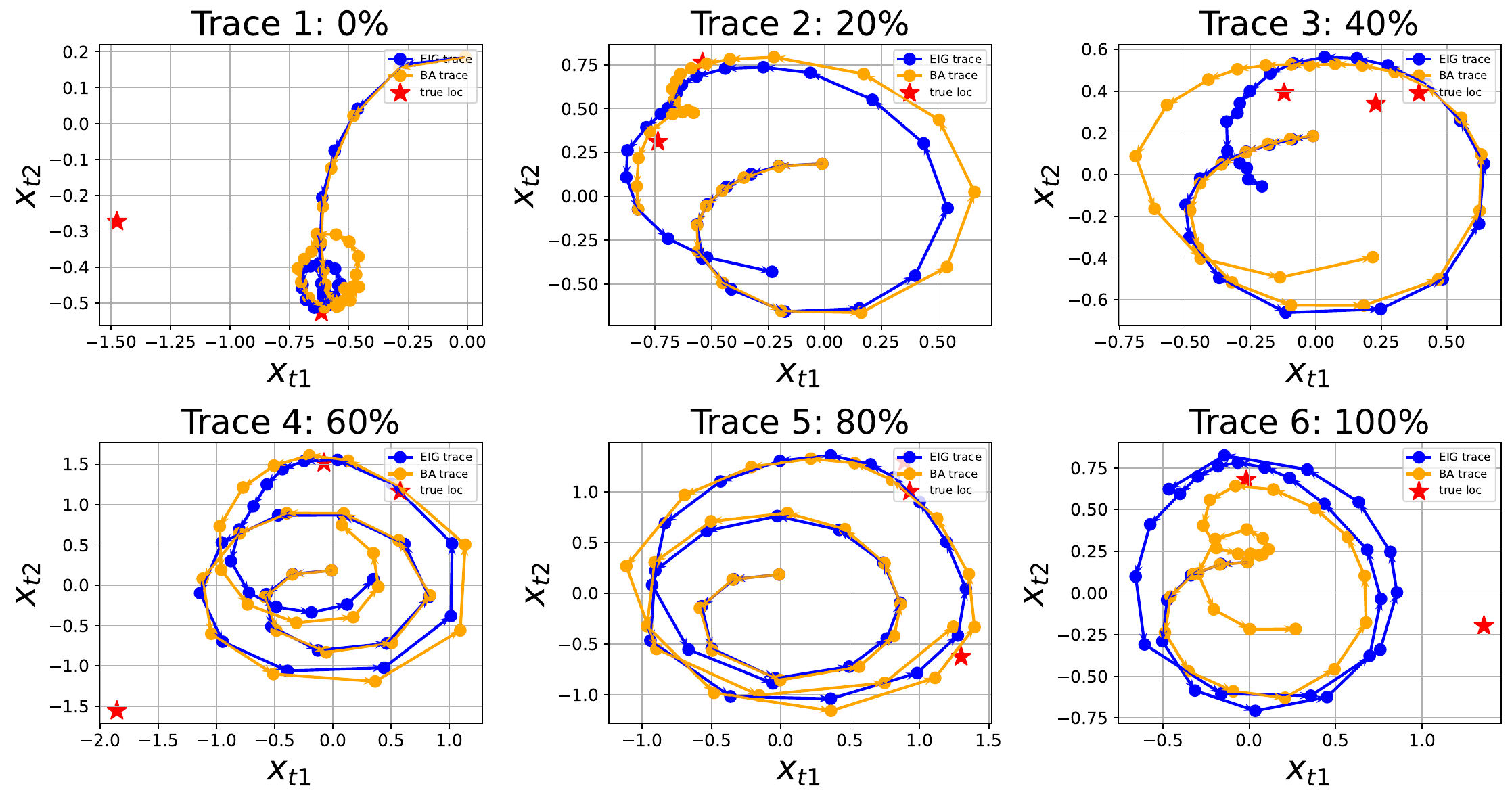}
    \caption{\small Representative design traces for the source-location example. For each sampled value of $\theta$, we run the EIG and BA policies repeatedly and summarize the resulting stochastic traces by geometric medians. The panels show parameter values for which the DTW distance between the representative traces corresponds to the selected quantiles $(0\%,20\%,40\%,60\%,80\%,100\%)$ over a sample from the prior. The red star-like dots are the true signal sources.  The blue arrows are representative design traces from the EIG policy, and the orange arrows are representative design traces from the BA policy.}
\label{fig design trace random: loc}
\end{figure}

\begin{figure}[!htbp]
    \centering
    \includegraphics[width=0.6\linewidth]{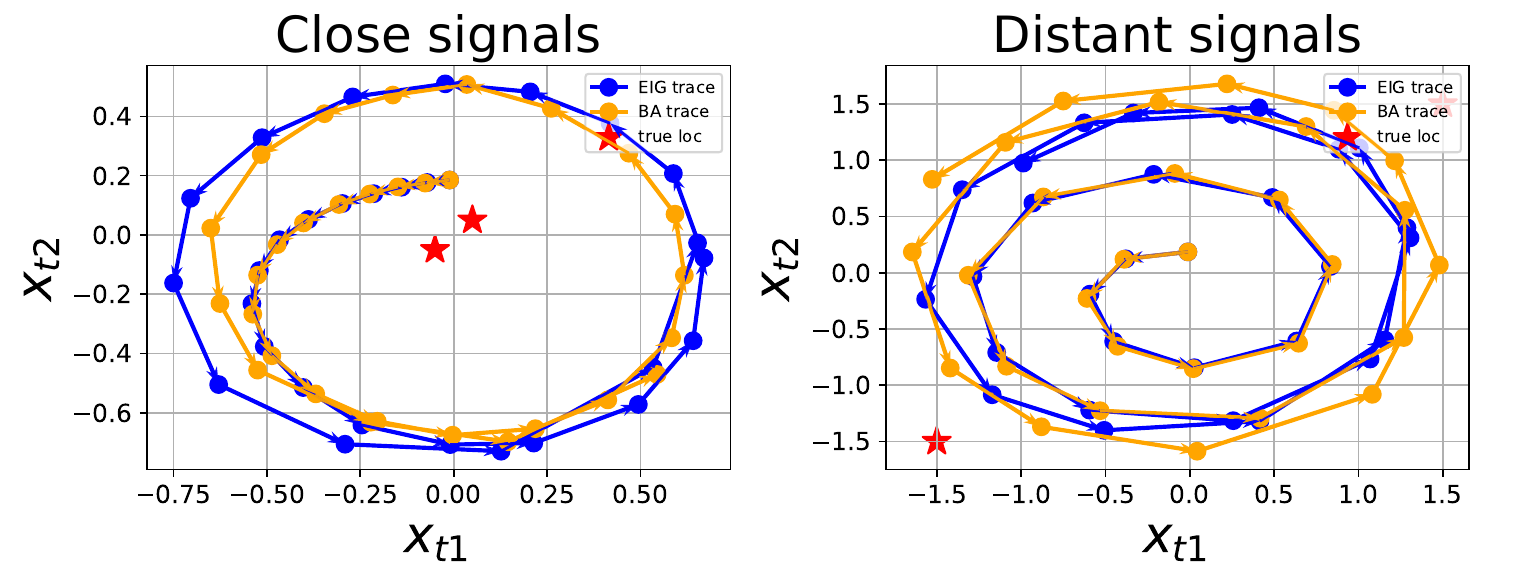}
    \caption{\small Representative design traces for the source-location example at manually selected parameter values.  For each value of $\theta$, we run the EIG and BA policies repeatedly and summarize the resulting stochastic traces by geometric medians.  The red star-like dots are the true signal sources.  The blue arrows are representative design traces from the EIG policy, and the orange arrows are representative design traces from the BA policy.}
\label{fig design trace manual: loc}
\end{figure}

\subsection{Hyperbolic temporal discounting}\label{subsec:hyperbolic}

Next we consider a design problem for a hyperbolic discounting model discussed in
\cite{foster+imr21}.  The purpose of experiments like this one is to study
how much people value immediate compared to delayed rewards.  
The example illustrates use of our approach in
a setting where the data are discrete.  Participants in a study were asked to choose
between the two propositions ``$\pounds R$ today'' or ``$\pounds 100$ in $D$ days'',
with values $V_0$ and $V_1$ respectively, where $V_0=R$ and $V_1=100/(1+k D)$ and
$k>0$ is a discounting parameter. The design variables are $x=(R,D)$, the immediate
reward and the delay.  We observe $y\in\{0,1\}$, equal to $1$ if the respondent chooses
the immediate reward $V_0$ rather than the delayed reward $V_1$, and it is assumed that
$$P(y=1\mid\theta,x)=\epsilon+(0.5-\epsilon) \left\{1+\operatorname{erf}\!\left(\frac{V_0-V_1}{a}\right)\right\},$$
where $\theta=(k,a)$ is the model parameter, $\operatorname{erf}$ is the error
function, $a>0$ controls the sharpness of the choice probability, and the lapse rate
is fixed at $\epsilon=0.01$, so that the choice probability ranges between $0.01$ and
$0.99$.  We place a log-normal prior on the discount parameter, $\log k\sim N(-4.25,1.5^2)$,
and a half-normal prior on the sharpness, $a \sim N_+(0,2^2)$, with $k$ and $a$
independent a priori.

In the policy-based formulation, the network does not output the design variables
$x=(R,D)$ directly. Instead it outputs two unconstrained real values which are
transformed onto the valid ranges of the reward and delay.  Writing
$(r,d)$ for the raw network outputs, the immediate reward is obtained as
$R=100\,F_{3}(r)$, where $F_{3}$ is the cumulative distribution function of a
Student-$t$ distribution with three degrees of freedom, mapping $r$ into the
interval $(0,100)$.  The delay is obtained as $D=\exp\left\{d-\mathrm{E}(\log k)\right\}$,
where $E(\log k)=-4.25$, which
centres the delay on a scale commensurate with the discount rate.  We consider, similar
to \cite{foster+imr21}, an adaptive design for this problem with a sequence of $T=20$ steps, choosing the design variables and making an observation $y_t$ at each step,
$t=1,\dots,T$.

\begin{figure}[!htbp]
    \centering
    \includegraphics[width=0.5\linewidth]{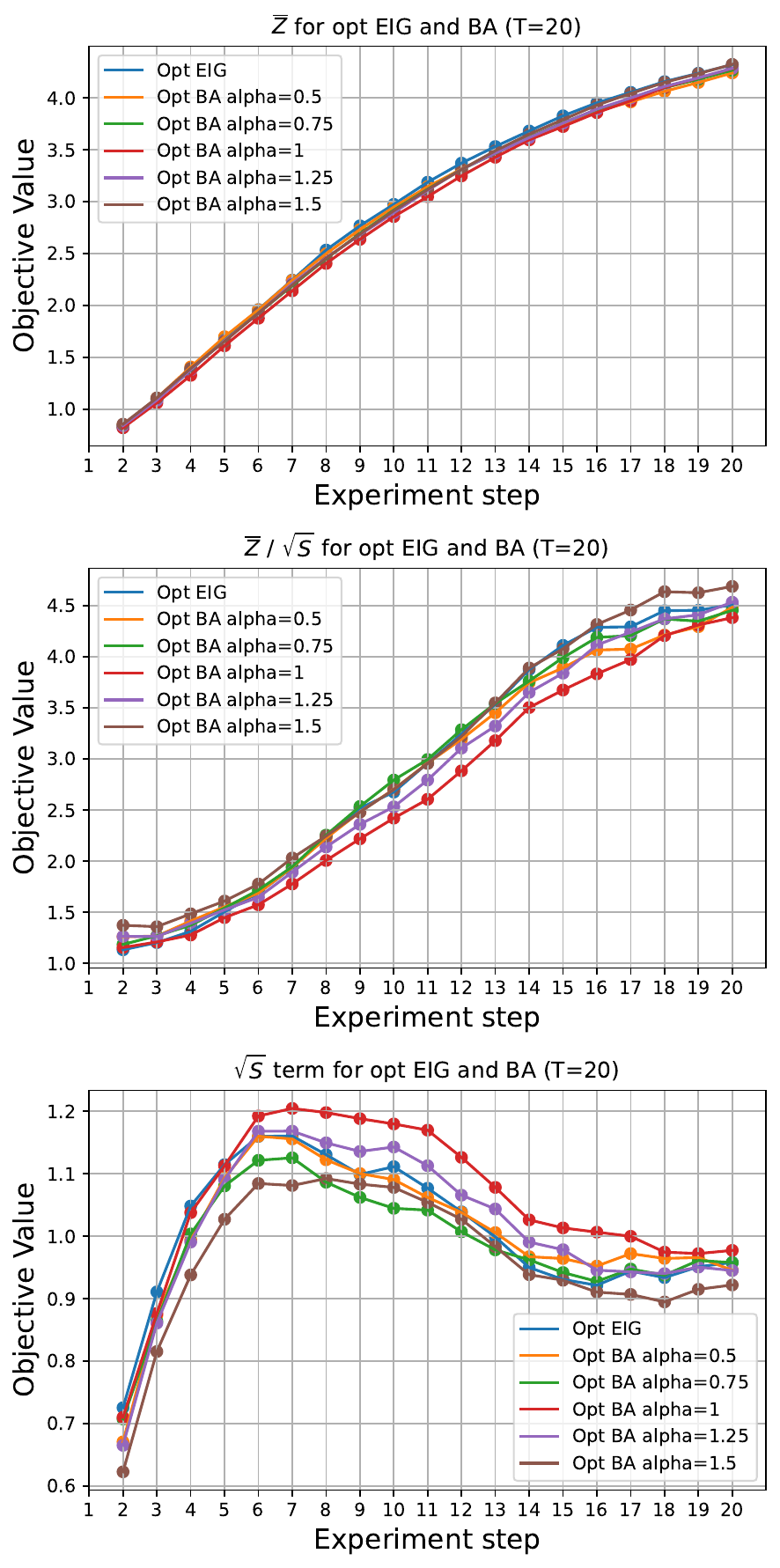}
    \caption{\small Evaluated $\overline{Z}$ (top), $\overline{Z}/\sqrt{S}$ (middle) and $\sqrt{S}$ (bottom) over $T=20$ experimental steps for hyperbolic
    discounting example for EIG policy and BA policies with $\alpha=0.5,0.75,1.0,1.25,1.5$.}
\label{fig loss: hyper}
\end{figure}

Figure~\ref{fig loss: hyper} compares the EIG policy with BA
policies trained at different values of $\alpha$. Similar to the location finding example, the top panel shows the Monte Carlo estimate $\overline Z(\varphi)$, the middle panel shows the ratio
$\widehat R(\varphi)=\overline Z(\varphi)/\sqrt{S(\varphi)}$ and the bottom panel shows the estimated standard deviation $\sqrt{S(\varphi)}$, all evaluated sequentially over the $T=20$ experimental steps.  The values of
$\overline Z(\varphi)$ increase steadily for all policies and are very similar
across the EIG and BA policies. The BA objective does not lead
to a substantial loss in the EIG lower-bound estimate.  The BA policies with
$\alpha=0.5,0.75$ and $1.5$ have smaller values of $\sqrt{S(\varphi)}$ over
most experimental steps, and this reduction in variability leads to larger
values of $\widehat R(\varphi)$, especially in the later stages of the
experiment.  For the BA policies, the expected log relative
belief ratio is similar to that of the EIG policy, but its variability is
reduced.  

\begin{figure}[!htbp]
    \centering
    \includegraphics[width=0.6\linewidth]{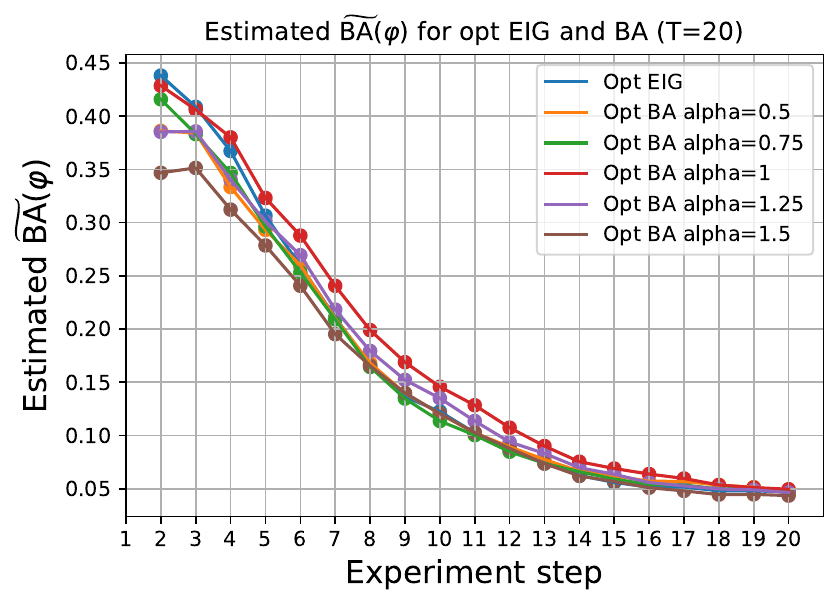}
    \caption{\small Estimated upper bound on bias against for the hyperbolic-discounting example. 
    The plotted quantity is Estimated $\tilde{BA}(\varphi) = \{1+\widehat R(\varphi)^2\}^{-1}$, where
    $\widehat R(\varphi)=\overline Z(\varphi)/\sqrt{S(\varphi)}$ is evaluated
    sequentially over the $T=20$ experimental steps.  Smaller values correspond
    to a tighter upper bound on the bias against.}
\label{fig tildeBA: hyper}
\end{figure}

The corresponding effect on the estimated bias-against upper bound is
shown directly in Figure~\ref{fig tildeBA: hyper}, which presents the same comparison on the scale of the
bias-against upper bound.  Since the bound is a decreasing transformation of
$\widehat R(\varphi)$, policies with larger values of the ratio criterion in
Figure~\ref{fig loss: hyper} correspond to smaller values in
Figure~\ref{fig tildeBA: hyper}. The BA policies reduce the estimated upper bound
on bias against compared to the EIG policy over the $T=20$ experimental steps.

\begin{figure}[!htbp]
    \centering
    \includegraphics[width=0.95\linewidth]{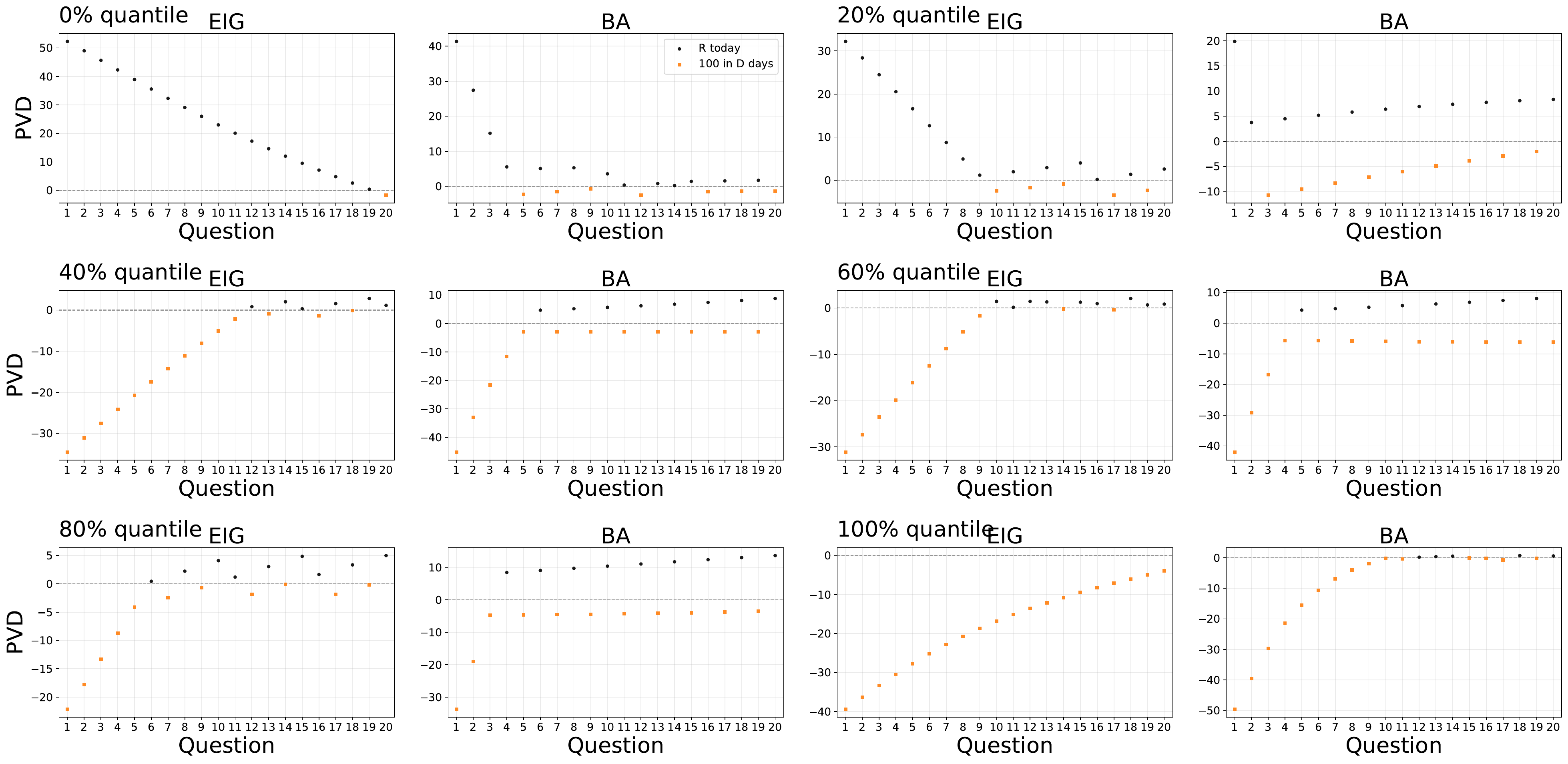}
    \caption{\small Representative subjective-value traces for the hyperbolic discounting example. The panels show parameter values for which the DTW distance between the representative traces corresponds to the selected quantiles $(0\%,20\%,40\%,60\%,80\%,100\%)$ over a sample from the prior.  Black dots indicate tasks for which ``$R$ pounds today'' has higher subjective value, while orange dots indicate tasks for which ``$100$ pounds in $D$ days'' has higher subjective value. The $y$-axis denotes the perceived value difference (PVD).}
\label{fig design trace random: hyper}
\end{figure}

\begin{figure}[!htbp]
    \centering
    \includegraphics[width=0.95\linewidth]{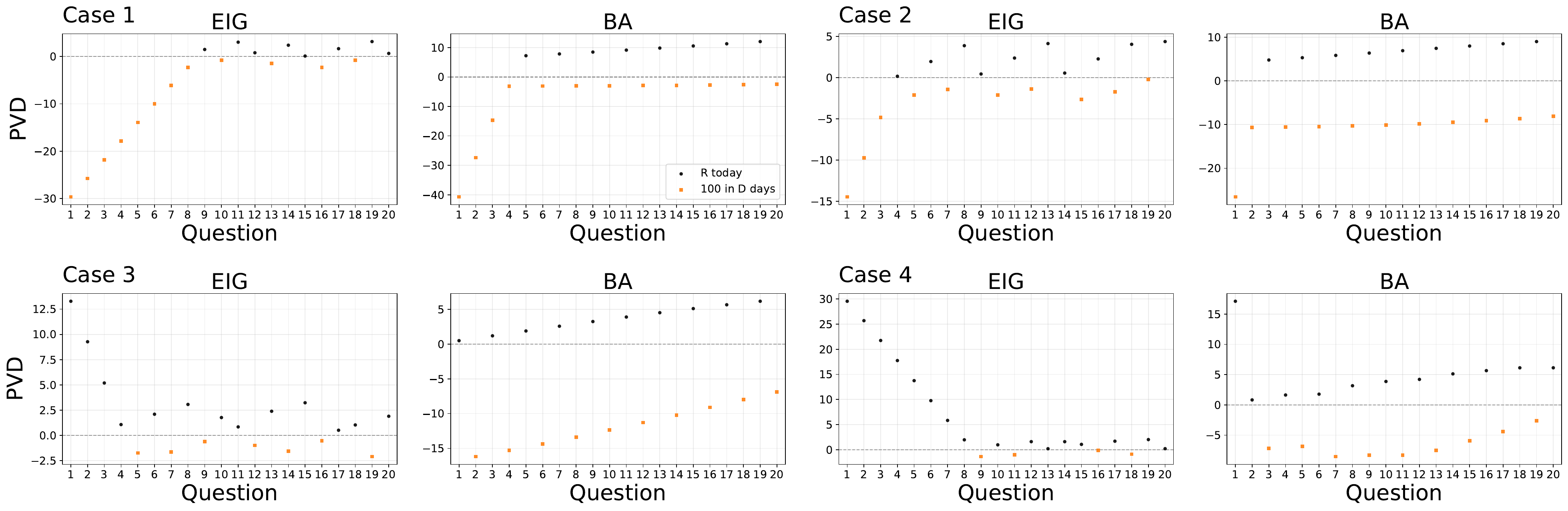}
    \caption{\small Representative subjective-value traces for the hyperbolic discounting example at manually selected parameter values. Black dots indicate tasks for which ``$R$ pounds today'' has higher subjective value, while orange dots indicate tasks for which ``$100$ pounds in $D$ days'' has higher subjective value. The $y$-axis denotes the perceived value difference (PVD). Parameter values $(k,a)$ for Case 1/2/3/4 (Top left/Top right/Bottom left/Bottom right): $(0.005,0.5)/(0.015,1.0)/(0.05,2.0)/(0.10,3.0)$.}
\label{fig design trace manual: hyper}
\end{figure}

Figures~\ref{fig design trace random: hyper} and
\ref{fig design trace manual: hyper} show the same comparison on the subjective
value scale.  For a proposed design $x_t=(D_t,R_t)$ and parameter value
$\theta=(k,a)$, the plotted quantity, the perceived value difference (PVD), is
\[
R_t-\frac{100}{1+kD_t}.
\]
If this is positive, it 
means that the immediate reward has the larger subjective value,
whereas if it is negative it means that the delayed reward has the larger subjective
value.  Values close to zero correspond to difficult questions, for which the
choice probability is close to one half.
In the plots, the BA policy generally shows larger changes in the perceived
value difference at consecutive experimental steps, often alternating between
design choices with positive and negative subjective value.

\subsection{A discrete choice experiment}\label{subsec:discrete-choice}

Finally, we consider an adaptive design problem for a stated-preference discrete choice experiment to infer an individual's optimal option profile and context-dependent choice behavior. The motivating setting here consists of a series of choice tasks in which participants select their most preferred option from three pizzas varying in spiciness and cheesiness. Within a choice task, pizza $j$'s profile is
\[
x_j=(x_{j1},x_{j2})^\top\in[0,1]^2,\qquad j\in\{1,2,3\}.
\]
where the two coordinates denote the degree of spiciness and cheesiness, respectively.  The full design is
\[
x=(x_1^\top,x_2^\top,x_3^\top)^\top\in[0,1]^6.
\]
Preference for a specific flavor is unlikely to be monotonic. Instead, it may peak at a moderate intensity and decline at either extreme. Therefore, we apply  an ideal-point model \citep{cooper1983two,xu2020simultaneous}, which assumes that an option's subjective value (utility, hereafter) increases if it approaches the ideal profile in attribute space. In addition, we relax the above assumption on a fully rational decision maker by incorporating context-dependent effects into the utility specification \citep{rooderkerk2011incorporating}. The most ubiquitous context effect is the attraction effect whereby an option becomes more attractive (i.e., higher utility) when a similar but inferior alternative exists.
% need refine after discussion with Chen.
To set up, the unknown parameter vector for our discrete choice model is
\[
\theta=(\mu_1,\mu_2,\beta,\gamma,\lambda),
\]
where $\mu=(\mu_1,\mu_2)^\top$ is a participant's ideal point (i.e., most preferred pizza profile),
$\beta>0$ controls sensitivity to distance from the ideal point, $\gamma$
controls the strength of the attraction effect, and $\lambda>0$ controls the
locality of the attraction effect. Therefore, pizza $j$'s utility consists of a baseline component determined by its squared Euclidean distance from the ideal point $\mu$, adjusted by the attraction effect $A_j(x,\theta)$ induced by the other options, which is
\[
V_j(x,\theta)
=
-\beta\|x_j-\mu\|_2^2
+
\gamma A_j(x,\theta),
\]
and the choice probabilities are given by the standard softmax function,
\[
P(y=j\mid x,\theta)
=
\frac{\exp\{V_j(x,\theta)\}}
{\sum_{\ell=1}^3\exp\{V_\ell(x,\theta)\}},
\qquad j\in\{1,2,3\}.
\]
Specifically, $A_j(x,\theta)$ captures the total attraction effect of the remaining options $k$ on option $j$, with each option’s influence decaying exponentially at a rate of $\lambda$ with its distance from $j$.
\[
A_j(x,\theta)
=
\sum_{k\neq j}
\exp\{-\lambda\|x_j-x_k\|_2^2\}
D_{jk}^{\mathrm{hard}}(x,\theta),
\]
where $D_{jk}^{\mathrm{hard}}(x,\theta)$ is the indicator that pizza $k$ triggers the attraction effect when evaluating the pizza $j$'s utility. Only when option $k$ is dominated by the target option $j$ on each attribute (i.e., further from the optimal value $\mu_i$) does its presence increase the attractiveness of option $j$.  The hard dominance indicator is therefore
\[
D_{jk}^{\mathrm{hard}}(x,\theta)
=
\prod_{a=1}^2
\mathbf 1\{(x_{ka}-\mu_a)^2-(x_{ja}-\mu_a)^2>0\}.
\]

In total, we use $T=12$ choice tasks for the experiment.  The prior ranges are
\[
\mu_1,\mu_2\sim U(0,1),\qquad
\log\beta\sim U(\log5,\log30),\qquad
\gamma\sim U(-1,1),\qquad
\log\lambda\sim U(\log5,\log80).
\]  Training follows
the common warm-start protocol described in Appendix~\ref{app: procedure hyper}. For policy net training here, we use a smooth policy-output transformation to avoid gradient loss from clipping the design to the closed box $[0,1]^6$; meanwhile, a sigmoid function (with a fixed scale parameter of $c = 30$) is used as a smooth
relaxation of the domination indicator $D^{hard}_{jk} (x;\theta)$ above. Further details are given in Appendix~\ref{app:gumbel}.

\begin{figure}[!htbp]
    \centering
    \includegraphics[width=0.5\linewidth]{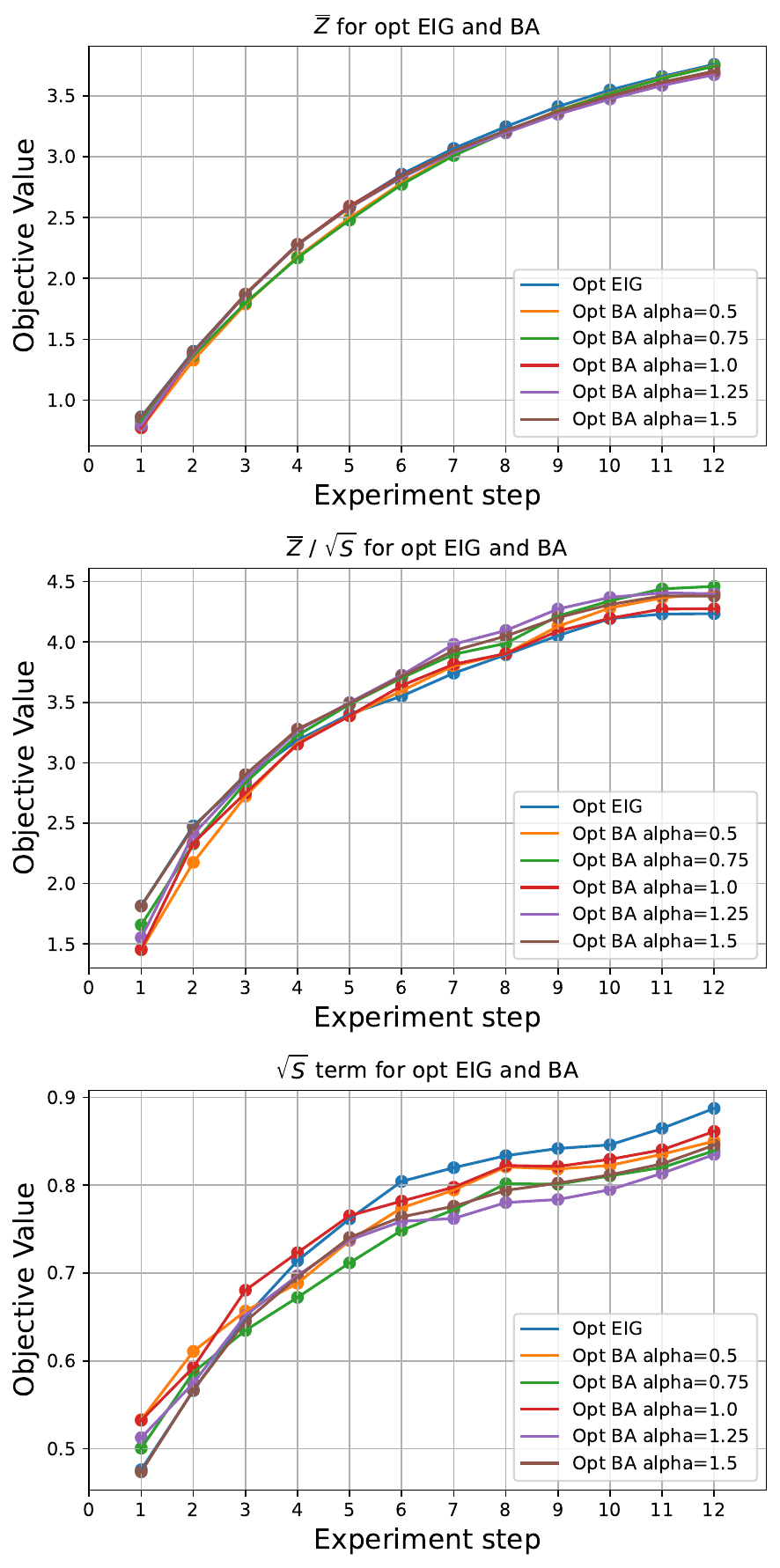}
    \caption{\small 
    Evaluated $\overline{Z}$ (top), $\overline{Z}/\sqrt{S}$ (middle) and $\sqrt{S}$ (bottom) over $T=12$ experimental steps for pizza example for EIG policy and BA policies with $\alpha=0.5,0.75,1.0,1.25,1.5$.}
\label{fig loss: pizza}
\end{figure}

Figure~\ref{fig loss: pizza} compares the EIG policy with BA policies
for the pizza choice example over $T=12$ choice tasks.  The top panel shows
that all policies achieve very similar values of $\overline Z(\varphi)$.  The
EIG policy is slightly higher in the later stages, but the gap is small, so that
training by the BA objective \eqref{eq:L-alpha-hat} does not substantially reduce the EIG lower-bound estimate. The main separation appears in the lower two panels. The BA policies, especially for $\alpha=1.5$, reduce $\sqrt{S(\varphi)}$ over most
experimental steps relative to the EIG policy.  This lower variability leads to
larger values of the ratio
$\widehat R(\varphi)=\overline Z(\varphi)/\sqrt{S(\varphi)}$, particularly
after the first few choice tasks.  Thus, in the pizza example, optimizing the BA objective gives a similar mean
log relative belief but smaller variability, again illustrating the design-dependent variance effect discussed in
Section~\ref{sec:eig-r-comparison}.  In terms of the ratio criterion, we choose
$\alpha=1.5$ as the BA policy for the design-trace evaluation. The difference between the learned policies is
stable across alternative combinations of Monte Carlo evaluation sample sizes; see
Figure~\ref{fig:robustness_pizza} in
Appendix~\ref{app: procedure hyper}.

\begin{figure}[!htbp]
    \centering
    \includegraphics[width=0.9\linewidth]{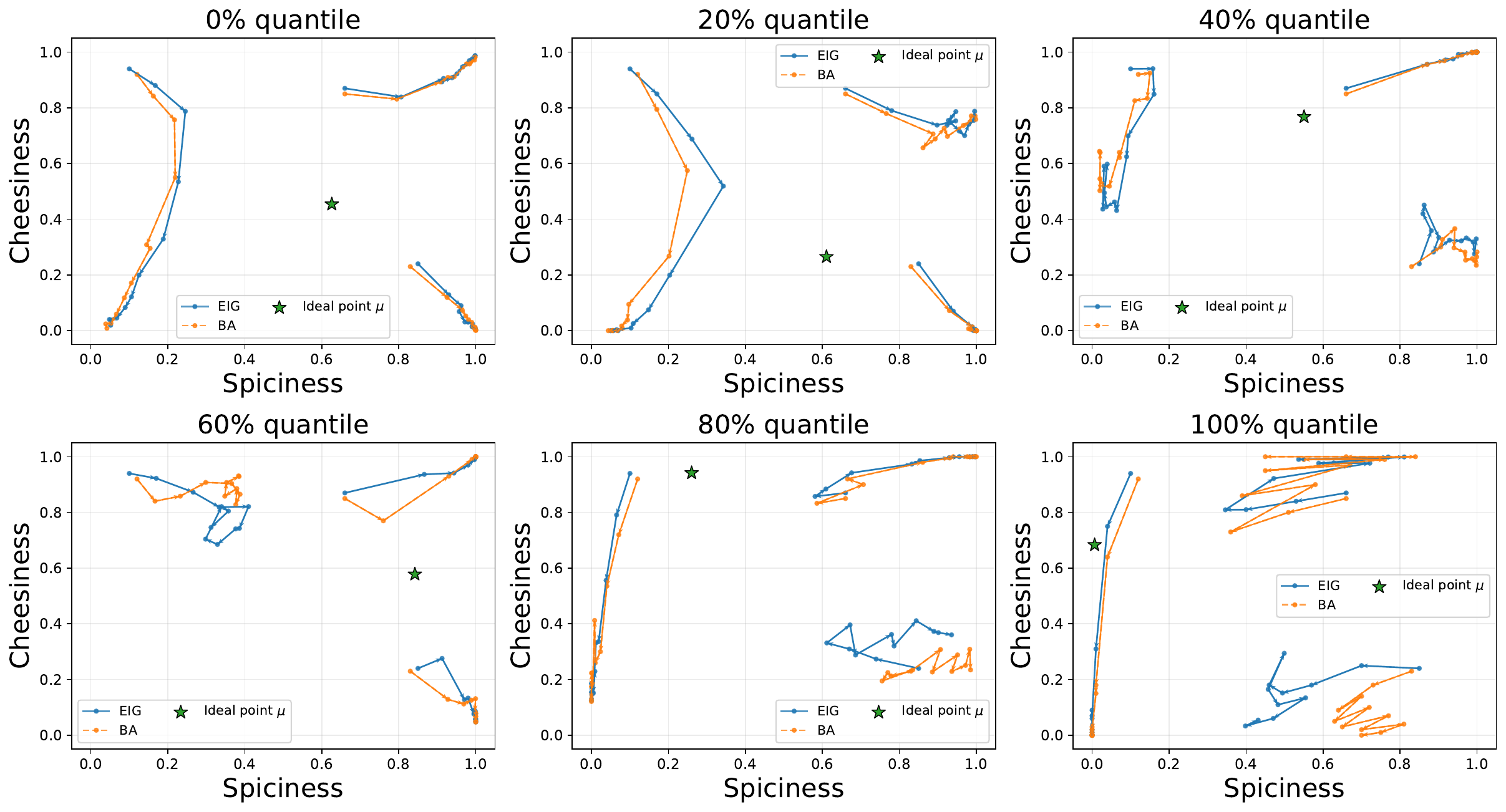}
    \caption{\small Representative design traces for the pizza choice example. For each sampled value of $\theta$, we run the EIG and BA policies repeatedly and summarize the resulting stochastic traces by geometric medians. The panels show parameter values for which DTW distance (summed over options) between the representative traces corresponds to the selected quantiles $(0\%,20\%,40\%,60\%,80\%,100\%)$ over a sample from the prior. The blue arrow curves are representative traces from the EIG policy, and the orange arrow curves are representative traces from the BA policy.  Each panel contains three curves, corresponding to the three pizza options.}
\label{fig design trace random: pizza}
\end{figure}

\begin{figure}[!htbp]
    \centering
    \includegraphics[width=0.6\linewidth]{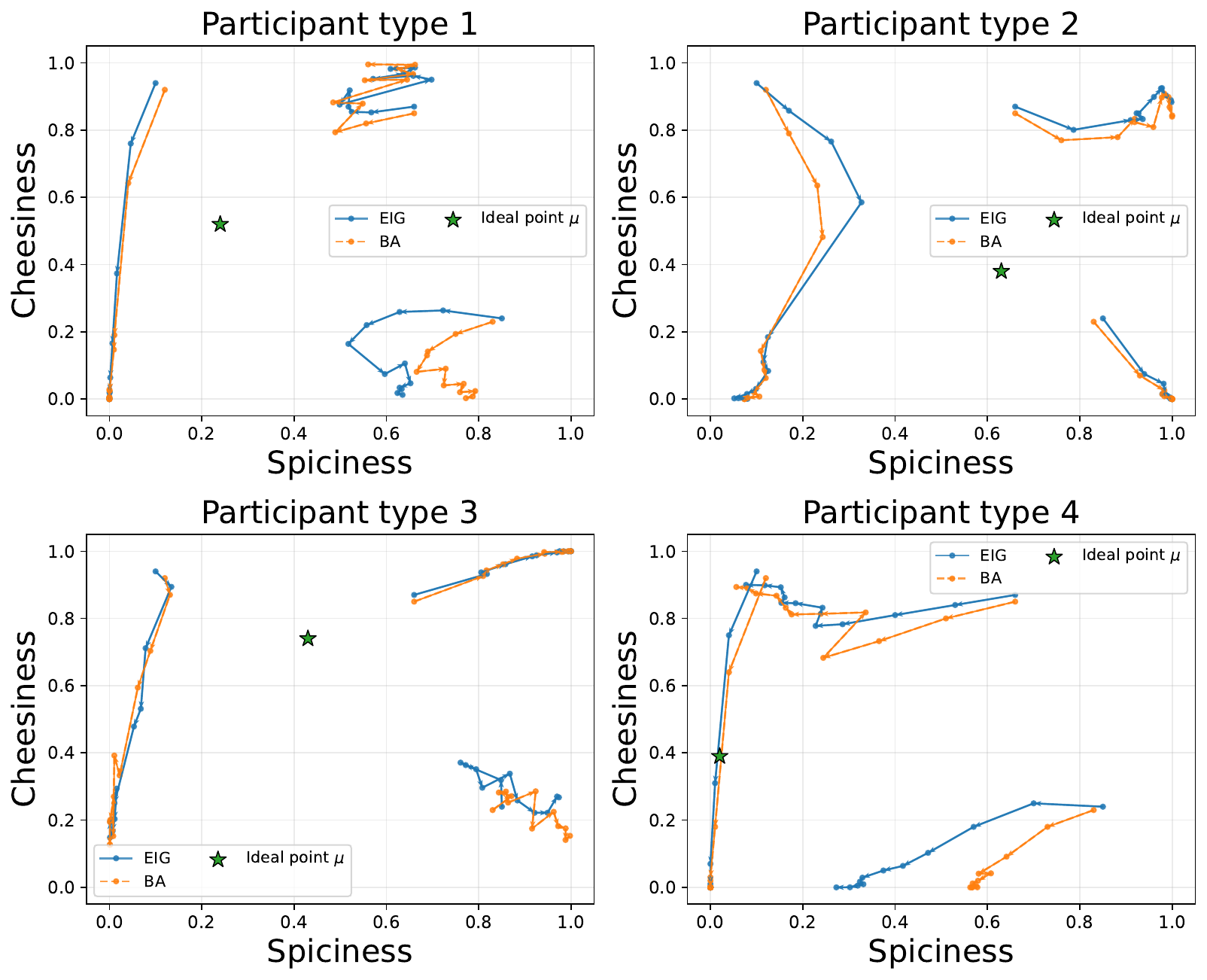}
    \caption{\small Representative design traces for the pizza choice example at manually selected parameter values.  For each value of $\theta$, we run the EIG and BA policies repeatedly and summarize the resulting stochastic traces by geometric medians for each option separately.  The blue arrow curves are representative traces from the EIG policy, and the orange arrow curves are representative traces from the BA policy.  Each panel contains three curves, corresponding to the three pizza options. Parameter values $(\beta,\gamma,\lambda)$ for Pizza customer participant type 1/2/3/4 (Top left/Top right/Bottom left/Bottom right): $(14.50,0.43,57.00)/(16.00,0.82,32.00)/(9.50,0.14,6.00)/(17.00,0.74,78.00)$.}
\label{fig design trace manual: pizza}
\end{figure}

Figures~\ref{fig design trace random: pizza} and
\ref{fig design trace manual: pizza} compare the representative pizza-design
traces produced by the EIG and BA policies. Each point in a trace is one pizza option design in the spiciness-cheesiness space, and the three separated
traces correspond to the adaptive design movement of the three pizza
options across the $T=12$ choice tasks. The ideal point $\mu$ is shown in
each panel, so the traces can be interpreted relative to the participant type being
simulated. The BA policy tends
to move pizza options attributes slightly differently from those of EIG policy. Usually one of the three pizza options would have its BA policy design closer to the corner. The BA policy changes the design behavior by emphasizing boundary choice sets, while
Figure~\ref{fig loss: pizza} shows that this is accompanied by a smaller
$\sqrt{S(\varphi)}$ and a similar value of $\overline Z(\varphi)$.

\begin{figure}[!htbp]
    \centering
    \includegraphics[width=0.55\linewidth]{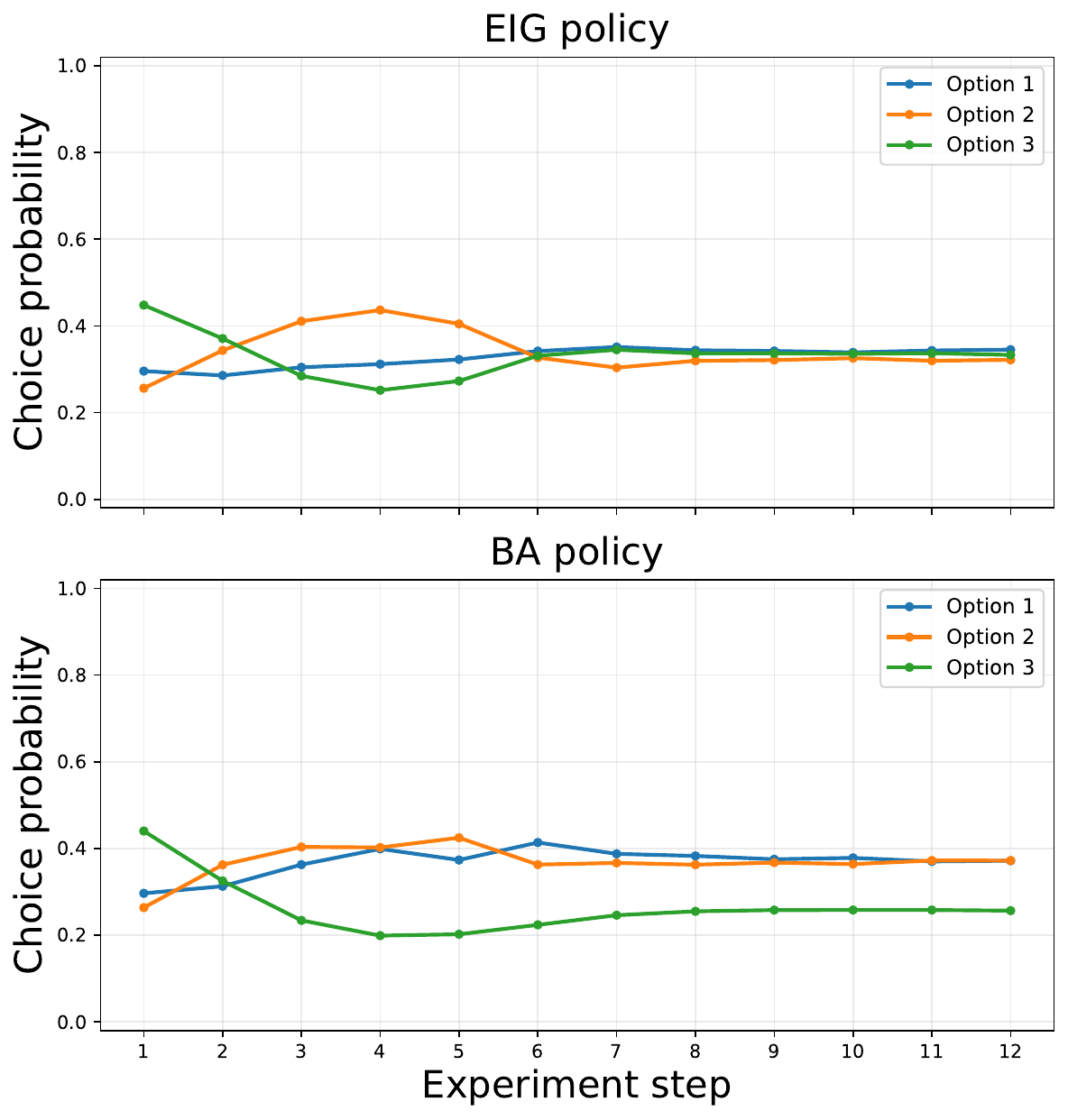}
    \caption{\small Average choice probabilities for the EIG and BA policies in the pizza choice example. For each parameter over a prior sample, evaluate the representative design traces and the corresponding choice probability at each experiment step. The averages of choice probabilities at each step are computed over the prior-sampled parameter values.}
\label{fig choice probability: pizza}
\end{figure}

\begin{table}[!htbp]
\centering
\small
\begin{tabular}{lcc}
\toprule
 & \makecell{Trained by EIG bound\\(Eq. \eqref{EIG-est12})} & \makecell{Trained by BA bound\\(Eq. \eqref{eq:L-alpha-hat})} \\
\midrule
\makecell{Freq. of\\domination of Option 1 } & 0.1703 & 0.1395 \\
\makecell{Freq. of\\domination of Option 2} & 0.1598 & 0.1343 \\
\makecell{Freq. of\\domination of Option 3} & 0.1773 & 0.3052 \\
\makecell{Freq. of having at\\least one dominated option} & 0.4182 & 0.4800 \\
\bottomrule
\end{tabular}
\caption{\small Frequency of attribute-wise dominated options in the pizza choice example, over prior-sampled parameter values and $T=12$ experimental steps. Option $j$ is counted as a dominated option when it is farther from the ideal point than another option on both spiciness and cheesiness. The last column gives the frequency with which at least one of the three options is dominated.}
\label{tab option decoy: pizza}
\end{table}

Figure~\ref{fig choice probability: pizza} and Table~\ref{tab option decoy: pizza} empirically illustrate the differences between the EIG and BA policies in terms of choice probabilities and choice contexts.
As shown in Figure~\ref{fig choice probability: pizza}, the EIG policy generates designs in which the alternatives become progressively more balanced in their attribute profiles and the probability of being chosen over the course of the experiment. This pattern reflects its objective of maximizing EIG, as it is generally more informative when multiple choice outcomes remain plausible and share almost the same probability of being chosen. In contrast, the BA policy tends to maintain one relatively dominated alternative over experiment steps, which has lower chosen probability. It is worth noting that including an inferior option in the experimental design does not necessarily reduce its informativeness in this case. Although the dominated option is less likely to be selected, its presence activates the attraction-effect term $\gamma A_j(\bm{x},\bm{\theta})$ to be non-zero in the utility function $V_j(\bm{x},\bm{\theta})$, thereby providing additional information for estimating the associated parameters $\gamma$ and $\lambda$. Conversely, if the experimental design does not deliberately include a dominated option, it may provide insufficient variation to identify the attraction-effect parameters, resulting in degenerate estimates such that $\hat{\gamma}=0$ and $\hat{\lambda}\to+\infty$. Hence, the BA policy suggests a mechanism that may improve identification of these parameters. 
Table~\ref{tab option decoy: pizza} further demonstrates the BA policy's advantages by showing that the frequency with which a dominated option appears across all tested experimental trajectories is 0.4800 under BA, compared with 0.4182 under EIG. As a result, EIG tends to generate relatively balanced alternatives, whereas BA deliberately includes a dominated alternative more frequently, achieving comparable mean information gain with lower variability in evidential performance.

\section{Discussion}

Our work has considered a design criterion based on the minimization of evidential bias, 
the bias against criterion, as an alternative to EIG.  While there is extensive existing research
on alternatives to EIG for batch design, they have been studied far less for amortized
sequential design.  Because the bias against is defined as a probability,
directly optimizing it is difficult, and so our strategy has been to control the bias against
by minimization of an upper bound based on the Paley-Zygmund inequality.  Optimization of
this bound can be formulated as maximization of a variance-penalized EIG criterion, which is
both convenient to optimize and interpretable.  Furthermore, the choice of penalty has a
natural interpretation as corresponding to controlling the bias against at different evidential
thresholds.  Amortized sequential design has been implemented using the deep adaptive design
approach of \citet{foster+imr21}, adapted to our new criterion.  In several examples, including
a complex discrete choice model, we are able to obtain designs which substantially improve on
EIG-based designs in terms of the bias against upper bound.

The main idea of our approach can be extended in several directions.
In the discussion of bias against in \citet[Section~4.6]{evans15}, a complementary concept
is also defined, termed the bias in favour.  While bias against is the probability of misleading
evidence for the truth, the bias in favour is the probability that evidence is produced against a
false value, with the false value chosen to differ from the truth by a practically significant amount.
An obvious topic for future work is to combine the bias against with the bias in favour
to develop design approaches that take into account which differences are practically significant
to detect in a given problem.

\section*{Acknowledgements}

David Nott's research was supported by the Ministry of Education, Singapore, under the Academic Research Fund Tier 2 (MOE-T2EP20123-0009).

\appendix

\section{Proof of Proposition~\ref{prop:PZ}}\label{app:PZ}

We reproduce the proof for completeness.  For $\delta\in[0,1]$, let $A=\{Z>\delta E(Z)\}$.  Writing 
$\mathbb{I}_A$ and $\mathbb{I}_{A^c}=1-\mathbb{I}_A$ for the corresponding 
indicators, 
\begin{align*}
  (1-\delta)E(Z) & = E(Z-\delta E(Z)) = E((Z-\delta E(Z))\mathbb{I}_A) + E((Z-\delta E(Z))\mathbb{I}_{A^c}) \\
  & \leq E((Z-\delta E(Z))\mathbb{I}_A),
\end{align*}
since $Z-\delta E(Z)\leq 0$ on $A^c$.  By the Cauchy--Schwarz inequality, 
\begin{align*}
  E((Z-\delta E(Z))\mathbb{I}_A) & \leq \sqrt{E((Z-\delta E(Z))^2)}\sqrt{P(A)}.
\end{align*}
Combining these and using $E(Z)\geq 0$ (so that $(1-\delta)E(Z)\geq 0$ and the 
inequality can be squared without changing direction), 
\begin{align*}
  (1-\delta)^2 E(Z)^2 & \leq E((Z-\delta E(Z))^2)\,P(A),
\end{align*}
Observing that $E((Z-\delta E(Z))^2)=\mathrm{Var}(Z)+(1-\delta)^2E(Z)^2$, it is
sufficient to have $\mathrm{Var}(Z)>0$ for this to be positive.  Then rearranging
the last inequality above gives \eqref{eq:PZ-lower}.  The complementary bound 
\eqref{eq:PZ-upper} follows since
\begin{align*}
  P(A^c) & = 1-P(A) \leq 1-\frac{(1-\delta)^2 E(Z)^2}{\mathrm{Var}(Z)+(1-\delta)^2 E(Z)^2}
  = \frac{\mathrm{Var}(Z)}{\mathrm{Var}(Z)+(1-\delta)^2 E(Z)^2}.
\end{align*}

\section{The bias-against bound and variance-penalized EIG}\label{app:Lagrangian}

Here we make precise the correspondence used in 
Section~\ref{sec:penalized}, namely that 
maximizing $R_t(x)$ is equivalent to maximizing the 
variance-penalized criterion at a suitable penalty.  
Throughout, write $\mu(x):=E(Z(y,x,\theta))=\text{EIG}(x)$ and 
$v(x):=\mathrm{Var}(Z(y,x,\theta))$ for a design $x$ in some class $\mathcal{X}$, 
and for a fixed constant $t$ define 
\begin{align*}
  R_t(x) := \frac{\mu(x)-t}{\sqrt{v(x)}}, \qquad 
  L_\alpha(x) := \mu(x)-\alpha\, v(x).
\end{align*}

The proof of Proposition 2 below 
is based on the quadratic transform for fractional programming 
\citep{shen+y18,shen+y25}, which for statisticians is naturally 
described using a minorization--maximization (MM) argument \citep{hunter+l04}.  
The minorization that is the basis of the proof derives from 
arguments in \citet[Section~VIII]{shen+y25}.  We introduce some 
notation to simplify the discussion.  
For $x$ with $v(x)>0$ write $\zeta(x):=R_t(x)^2=(\mu(x)-t)^2/v(x)$, and for a point $z$ with $v(z)>0$ define the surrogate 
\begin{align}
  g(x\mid z) & := 2\kappa(z)(\mu(x)-t)-\kappa(z)^2\,v(x), \qquad 
  \kappa(z):=\frac{\mu(z)-t}{v(z)}. \label{eq:mm-surrogate}
\end{align}
Since 
\begin{align}
  \zeta(x)-g(x\mid z) & = \frac{\big[(\mu(x)-t)-\kappa(z)\,v(x)\big]^2}{v(x)}\;\geq\; 0,
  \label{eq:minorize}
\end{align}
with equality at $x=z$, the surrogate minorizes $\zeta$ and is tangent to it at 
$z$.  An MM 
step for maximizing $\zeta$ is
$z\leftarrow\arg\max_x g(x\mid z)$, and since 
$g(x\mid z)=2\kappa(z)[L_{\kappa(z)/2}(x)-t]$ with $\kappa(z)>0$ for $\mu(z)>t$, the MM 
step maximizes the variance-penalized criterion 
$L_\alpha$ at $\alpha=\kappa(z)/2$.  Proposition~\ref{prop:Lagrangian} states that a 
maximizer of the ratio $R_t$ is a fixed point of this iteration, with the 
surrogate there being $L_{\alpha^*}$.  No concavity of $\mu$ or 
convexity of $v$ is required here, although that is often assumed in fractional
programming problems.  These assumptions are needed there to make 
the surrogate objective concave in the optimization steps and to guarantee
convergence to a global maximizer.  However, these assumptions are not necessary
for the fixed-point property that Proposition 2 establishes.  

To give some further intuition for Proposition 2 it can be helpful
to consider what happens when $\mu$ and $v$ are differentiable.  In that case, 
at an interior maximizer of $x^*$ of $R_t(x)$, 
$$0=\nabla_x R_t(x^*)=v(x^*)^{-1/2} \left(\nabla_x \mu(x^*)-\frac{\mu(x^*)-t}{2v(x^*)} \nabla_x v(x^*)\right),$$
and the term in brackets is the stationarity condition at $x^*$ for
the penalized criterion $L_\alpha(x)$ with 
$\alpha=\alpha^*=(\mu(x^*)-t)/(2v(x^*))$.  
The proof of Proposition 2 requires no differentiability assupmptions
about $\mu$ and $v$.  

\begin{proposition}\label{prop:Lagrangian}
Suppose $x^*\in\mathcal{X}$ maximizes $R_t$ over $\{x\in\mathcal{X}:v(x)>0\}$, 
and that $\mu(x^*)>t$ and $v(x^*)>0$.  Set 
\begin{align}
  \alpha^* & = \frac{\mu(x^*)-t}{2\,v(x^*)}. \label{eq:alpha-star-app}
\end{align}
Then $x^*$ maximizes $L_{\alpha^*}$ over $\{x\in\mathcal{X}:v(x)>0\}$, and 
\begin{align}
  \max_{x:\,v(x)>0} L_{\alpha^*}(x) & = L_{\alpha^*}(x^*) = t+\frac{\mu(x^*)-t}{2}=\frac{\mu(x^*)+t}{2}. \label{eq:opt-L-alpha-star}
\end{align}
\end{proposition}

\begin{proof}
Write $r^*:=R_t(x^*)=(\mu(x^*)-t)/\sqrt{v(x^*)}>0$ and 
$\kappa^*:=\kappa(x^*)=(\mu(x^*)-t)/v(x^*)=2\alpha^*>0$, and note that
\begin{align}
  L_{\alpha^*}(x^*) & = \mu(x^*)-\alpha^* v(x^*)
    = \mu(x^*)-\frac{\mu(x^*)-t}{2} = t+\frac{\mu(x^*)-t}{2},
  \label{eq:opt-value}
\end{align}
which is the same as \eqref{eq:opt-L-alpha-star}.  We show 
$L_{\alpha^*}(x)\leq L_{\alpha^*}(x^*)$ for every $x$ with $v(x)>0$.  
There are two cases.  

If $\mu(x)\leq t$, then since $\alpha^* v(x)>0$ (since $\mu(x^*)>t$ implies $\alpha^*>0$, and $v(x)>0)$ 
we have
\begin{align*}
  L_{\alpha^*}(x) & = \mu(x)-\alpha^* v(x) < \mu(x) \leq t 
  < t+\frac{\mu(x^*)-t}{2} = L_{\alpha^*}(x^*),
\end{align*}
using $\mu(x^*)>t$ and \eqref{eq:opt-value}.

If $\mu(x)>t$, then $R_t(x)>0$, and optimality of $x^*$ gives 
$0<R_t(x)\leq r^*$. Squaring both sides of the last inequality,
\begin{align}
  \zeta(x) & \leq (r^*)^2=\zeta(x^*). \label{eq:phi-opt}
\end{align}
Applying the minorization \eqref{eq:minorize} at $z=x^*$, and using 
\eqref{eq:phi-opt} and the tangency of the surrogate at $x^*$,
\begin{align*}
  4\alpha^*\big[L_{\alpha^*}(x)-t\big] & = g(x\mid x^*) 
  \;\leq\; \zeta(x) \;\leq\; \zeta(x^*) = g(x^*\mid x^*) 
  = 4\alpha^*\big[L_{\alpha^*}(x^*)-t\big],
\end{align*}
where the first and last equalities use 
$g(x\mid x^*)=2\kappa^*[L_{\kappa^*/2}(x)-t]=4\alpha^*[L_{\alpha^*}(x)-t]$ and 
$g(x^*\mid x^*)=\zeta(x^*)$ follows from \eqref{eq:minorize} at 
$x=z=x^*$.  Dividing by $4\alpha^*>0$ gives 
$L_{\alpha^*}(x)\leq L_{\alpha^*}(x^*)$.

Combining the two cases, $x^*$ maximizes $L_{\alpha^*}$ over 
$\{x\in\mathcal{X}:v(x)>0\}$, with optimal value \eqref{eq:opt-value}.
\end{proof}

\begin{remark}\label{rem:policy-form}
Proposition~\ref{prop:Lagrangian} is stated for optimization over a design 
$x$, but its proof uses no structure of the set $\mathcal{X}$ beyond the 
values of the two functionals $\mu$ and $v$.  The result therefore applies 
in the sequential adaptive setting of Section~\ref{sec:computation}, 
with $x$ replaced by the policy parameter $\varphi$, $\mathcal{X}$ by the 
policy parameter space, and 
$\mu(\varphi)=E\{Z(y,\varphi,\theta)\}$ and 
$v(\varphi)=\mathrm{Var}\{Z(y,\varphi,\theta)\}$ the mean and variance of 
the log relative belief ratio of the sequential experiment, provided a 
maximizer $\varphi^*$ of $R_t$ over the policy parameter space exists.  
\end{remark}

\section{Details of derivations in Example 1} \label{app:example1}

\subsection{Derivations of the inequalities \eqref{xi0condition} and \eqref{xi1condition} in Example 1}

Here we derive \eqref{xi0condition} and \eqref{xi1condition} in the main 
text.  

For \eqref{xi0condition}, we first write
$I=\frac{1}{2}\log(1+\tau^2/s^2)$, 
which is the EIG conditional on $\xi=0$, and $\eta=y/(s^2+\tau^2)^{1/2}$.  Using the
expression \eqref{Kexpression} we can then write
\begin{align*}
    P(K>0|\xi=1) & = P\left(\frac{(y-\theta)^2}{s^2}<2I+\eta^2\Bigm| \xi=1\right).
\end{align*}
Given $\xi=1$, $y$ and $\theta$ are independent, with $\theta|\xi=1\sim N(0,\tau^2)$ and $y|\xi=1\sim N(0,\sigma_B^2)$, and so
\begin{align*}
  P(K>0|\xi=1) & = E_{y\sim p(y|\xi=1)}\left(P\left(\frac{(y-\theta)^2}{s^2}<2I+\eta^2\Bigm| y, \xi=1\right)\right),
\end{align*}
where $(y-\theta)/s|y,\xi=1\sim N(y/s,\tau^2/s^2)$, with this normal density being
upper bounded by $s (2\pi \tau^2)^{-1/2}$.  Multiplying the upper bound by the
width of the interval $(-(2I+\eta^2)^{1/2},(2I+\eta^2)^{1/2})$, we obtain
a bound on the probability
\begin{align*}
    P(K>0|\xi=1) & \leq E_{y\sim p(y|\xi=1)}\left(2s(2\pi\tau^2)^{-1/2}(2I+\eta^2)^{1/2}\right) \\
    & \leq E_{y\sim p(y|\xi=1)}\left(2s(2\pi\tau^2)^{-1/2}((2I)^{1/2}+|\eta|)\right), \\
    & \leq Cs(\log n)^{1/2}
\end{align*}
for some constant $C$ depending only on $\tau$ and $\sigma_B$, 
where in the second to last line we have used $(2I+\eta^2)^{1/2}\leq ((2I)^{1/2}+|\eta|)$ and in the last line we have used 
$y|\xi=1\sim N(0,\sigma_B^2)$ and $I=1/2\log n+O(1)$.  Hence
$P(K\leq 0|\xi=1)\geq 1-Cs(\log n)^{1/2}$.

For \eqref{xi1condition}, we have 
\begin{align*}
    P(K\leq 0|\xi=0) & \leq P\left( \log\left(1+\frac{\tau^2}{s^2}\right)\leq W^2\Bigm| \xi=0\right),
\end{align*}
where $W=(y-\theta)/s\sim N(0,1)$ when $\xi=0$.  Now $W^2|\xi=0$ is $\chi^2_1$ with 
moment generating function $M(t)=(1-2t)^{-1/2}$ for $|t|<0.5$.  
Applying a Chernoff bound with $t=1/4$ and noting that $M(1/4)=\sqrt{2}$ gives 
\begin{align*}
    P(K\leq 0 |\xi=0) \leq 
    & \sqrt{2}\left(1+\frac{\tau^2}{s^2}\right)^{-1/4}.
\end{align*}
Noting that 
$\left(1+\frac{\tau^2}{s^2}\right)^{-1/2}=\left(\frac{s^2}{s^2+\tau^2}\right)^{1/2}\leq \left(\frac{s^2}{\tau^2}\right)^{1/2}=s/\tau,$
we see that
\begin{align*}
  P(K\leq 0 |\xi=0)\leq \sqrt{2}(s/\tau)^{1/2},
\end{align*}
which demonstrates \eqref{xi1condition}.  

\subsection{Derivation of the inequality \eqref{EIG-inequality} in Example 1}
 
We need some notation first.  Let $X$, $Y$ and $Z$ be random variables.  
The mutual information of $X$ and $Y$ is
$$I(X;Y):=\KL{p(X,Y)}{p(X)p(Y)}.$$
The EIG is $I(\theta;y)$ which follows directly from \eqref{EIG1}.  
The conditional mutual information $I(X;Y|Z)$ is
$$I(X;Y|Z):=E_{Z\sim p(Z)}\left(\KL{p(X,Y|Z)}{P(X|Z)p(Y|Z)}\right).$$
The chain rule for mutual information implies that
\begin{align*}
  I(\theta;y,\xi)=I(\theta;y)+I(\theta;\xi|y),
\end{align*}
which upon rearranging and using the property that mutual information and 
conditional mutual information are non-negative implies
\begin{align}
  I(\theta;y,\xi)-I(\theta;\xi|y)= I(\theta;y)=\mathrm{EIG}\leq I(\theta;y,\xi). \label{EIG-inequality-prelim} 
\end{align}
Now, 
\begin{align}
  I(\theta;y,\xi) & =\KL{p(\theta,y,\xi)}{p(\theta)p(y,\xi)} \nonumber \\
  & = E_{\xi\sim p(\xi)}\left(\KL{p(\theta,y|\xi)}{p(\theta)p(y|\xi)}\right) \nonumber \\
  & = (1-q)\KL{p(\theta,y|\xi=0)}{p(\theta)p(y|\xi=0)}+q\KL{p(\theta,y|\xi=1)}{p(\theta)p(y|\xi=1)}. \label{I-theta-y-xi}
\end{align}
Since $\theta$ and $y$ are independent given $\xi=1$, 
the second term in the last line above is zero.  The first term 
in \eqref{I-theta-y-xi} is the conditional mutual information between $\theta$ and $y$ given $\xi=0$ (i.e. the $\mathrm{EIG}$
for the experiment where $\xi=0$), the
normal location problem, and this can be computed analytically using conjugacy
so that we have
\begin{align}
  I(\theta;y,\xi)=\frac{1-q}{2}\log\left(1+\frac{\tau^2}{s^2}\right). \label{term1}
\end{align}

Next, consider $I(\theta;\xi|y)$.  We have
\begin{align}
    I(\theta;\xi|y)& = E_{y\sim p(y)}\left(\KL{p(\theta,\xi|y)}{p(\theta|y)p(\xi|y)}\right) \nonumber \\
    & = E_{(y,\theta)\sim p(y,\theta)}\left(E_{\xi\sim p(\xi|y,\theta)}(\log p(\xi|y,\theta)\right)-E_{y\sim p(y)}\left(E_{\xi\sim p(\xi|y)}(\log p(\xi|y))\right) \nonumber 
\end{align}
Recall from Section~2.3 the notation for the Shannon or differential 
entropy of a random
variable $X$ as
$H(X)=-E_{X\sim p(X)}(\log p(X))$.  If $X$ is binary with
$P(X=1)=p$, then $H(X)=-p\log p -(1-p)\log (1-p)$, and abusing notation
we also write $H(p):=H(X)$.    
Then we have
\begin{align}
    I(\theta;\xi|y)
    & = E_{y\sim p(y)}(H(\xi|y))-E_{(y,\theta)\sim p(y,\theta)}(H(\xi|y,\theta))
     \nonumber \\
    & \leq E_{y\sim p(y)} (H(\xi|y)),   
\end{align}
where in the last line we have used the fact that Shannon entropy
for a discrete variable is non-negative.  
The function $H(\cdot)$ is concave, so by Jensen's inequality
\begin{align}
  I(\theta;\xi|y) & \leq H(E_{y\sim p(y)}(P(\xi=1|y)) \nonumber \\
  & = H(q). \label{term2}
\end{align}
Hence from \eqref{EIG-inequality-prelim}, \eqref{term1} and \eqref{term2} we obtain
\begin{align*}
    \frac{1-q}{2}\log\left(1+\frac{\tau^2}{s^2}\right)-H(q)\leq 
    \mathrm{EIG}\leq \frac{1-q}{2}\log\left(1+\frac{\tau^2}{s^2}\right),
\end{align*}
and recalling that $s^2=\sigma^2/n$, we see that
$\mathrm{EIG}=\frac{1-q}{2}\log n + O(1)$.  This demonstrates that $\mathrm{EIG}\rightarrow\infty$. 

\section{Differentiable relaxations for discrete examples} \label{app:gumbel}

The hyperbolic-discounting and pizza choice examples both involve discrete
simulated outcomes.  This means that the reparametrization approach we have
used to construct differentiable approximations to the design objective
breaks down, since the simulated data is not a differentiable function 
of the policy parameters and noise.  
This occurs, for example, if the response is sampled as an exact Bernoulli or categorical variable. To obtain a differentiable approximation of the
training objective, we use the so-called Gumbel-softmax trick
\citep{jang2016categorical,maddison2016concrete} to relax discrete
distributions to continuous approximations.

An exact sample from a categorical distribution on $\{1,\dots, K\}$
with respective probabilities $p_1,\ldots,p_K$ can be obtained  
by first drawing independent standard Gumbel variables
\[
g_j=-\log\{-\log u_j\},\qquad u_j\sim \mathrm{Uniform}(0,1),
\]
$j=1,\dots, K$, and then setting
\[
y=\arg\max_{1\leq j\leq K}\{\log p_j+g_j\}.
\]
Below we write $e_y=(e_{y1},\dots, e_{yK})$ for the one-hot encoding of $y$ where
$e_{yj}=1$ if $y=j$ and $e_{yj}=0$ otherwise.  

The Gumbel-softmax relaxation of the distribution of $e_y$ 
replaces the random variable $e_{yj}$ with the continuous
random variable
\[
\widetilde y_j
=
\frac{\exp\{(\log p_j+g_j)/\tau\}}
{\sum_{\ell=1}^K \exp\{(\log p_\ell+g_\ell)/\tau\}},
\qquad j=1,\ldots,K.
\]
where $\tau>0$ is a temperature parameter.  
The distribution of $\widetilde{y}=(\widetilde{y}_1,\dots, \widetilde{y}_K)$
approximates the distribution of $e_y$, with a draw of $\widetilde{y}$ 
becoming closer to a draw of $e_y$ as $\tau\rightarrow 0$, and larger
values of $\tau$ giving a smoother density for $\widetilde{y}$.  
In training, samples of $\widetilde{y}$ are used to approximate the design
objective, and passed
to the policy network input-history.  

We use the Gumbel-softmax draw as a differentiable version of the
simulated categorical response.  The usual categorical log likelihood  is
\[
\log p_y=\sum_{j=1}^K e_{yj}\log p_j .
\]
Since the continuous simulator replaces $e_y$ by
$\widetilde y$, the differentiable log-likelihood approximation is
\[
\sum_{j=1}^K \widetilde y_j \log p_j.
\]

In the hyperbolic-discounting example of subsection \ref{subsec:hyperbolic} the response is binary.  For the binary response, let
$\psi(x,\theta)$ denote the model probability of a $1$ at
design $x$.  During training we replace the Bernoulli draw by
\[
\widetilde y \in (0,1),
\]
where abusing notation $\widetilde{y}$ is what we have called
$\widetilde{y}_2$ for the case of a categorical distribution on $\{0,1\}$ 
with $K=2$ in our previous discussion.  The log likelihood
term in the objective is approximated by
\[
\widetilde y \log \psi(x,\theta)
+
(1-\widetilde y)\log\{1-\psi(x,\theta)\}.
\]
Thus the simulator used in training is differentiable with respect to the
design through $\psi(x,\theta)$, and the usual reparameterized EIG and BA
objectives can be optimized by backpropagation.

In the pizza choice example of subsection \ref{subsec:discrete-choice}, the response is a three-option categorical choice. Let
\[
p_j(x,\theta)
=
P(y=j\mid x,\theta),
\qquad j=1,2,3,
\]
be the choice probabilities for the three pizza options.  During training we
draw $\widetilde y=(\widetilde y_1,\widetilde y_2,
\widetilde y_3)$ using
\[
\widetilde y_j
=
\frac{\exp\{(\log p_j(x,\theta)+g_j)/\tau\}}
{\sum_{\ell=1}^3 \exp\{(\log p_\ell(x,\theta)+g_\ell)/\tau\}}.
\]
The log-likelihood contribution is then approximated by
\[
\sum_{j=1}^3 \widetilde y_j \log p_j(x,\theta).
\]
which is the differentiable analogue of the usual categorical log likelihood.

The Gumbel-softmax relaxation is used only for training the policies.  For final
design-trace evaluation and interpretation of the learned policies, we return to
the intended discrete distributions, using Bernoulli responses in the hyperbolic example
and discrete categorical responses in the pizza choice example. For the hyperbolic-discounting example and pizza choice examples, we use $\tau=0.2$ and $0.5$ respectively in the Gumbel-softmax distribution. 

Especially for the pizza case study, two additional steps for smoothness and differentiability are required. First, for policy net training, the data simulator uses a soft domination function $D_{jk}^{\mathrm{soft}}(x,\theta)$  
to ensure differentiability and smoothness. Specifically,  an independent standard sigmoid function $\sigma(.)$ is applied on each attribute dimension such that
\[
D_{jk}^{\mathrm{soft}}(x,\theta)
=
\prod_{a=1}^2
\sigma\!\left(
c\{(x_{ka}-\mu_a)^2-(x_{ja}-\mu_a)^2\}
\right),
\]
where the scale parameter $c>0$ controls the sharpness with which the dominance between the compared options is perceived.  The soft dominance function $D_{jk}^{\mathrm{soft}}(x,\theta)$ constructs a smooth and differentiable data simulator for BOED, which converges to the strict dominance indicator as $c\to\infty$. Therefore, we fix a sufficiently large $c$, $c=30$, in the discrete choice experiment. 

Lastly, after policy net training, differentiability is also required in the mapping
from the policy-network output to the design space. We defined the design
space as the closed box $[0,1]^6$, while $0$ means no spiciness or cheesiness and $1$ means extreme spiciness or cheesiness. If this constraint is imposed by clipping
the raw policy output, then any coordinate clipped to $0$ or $1$ has zero local
gradient through the clipping operation.  To avoid this during training, the
policy first outputs an unconstrained raw vector
\[
z_t=\eta_\varphi({\cal I}_{t-1})\in\mathbb R^6 .
\]
This vector is arranged in a centre-and-offset form,
\[
\widetilde x_t =
\left(
z_{t1},
z_{t2},
z_{t1}+r_{\mathrm{loc}}z_{t3},
z_{t2}+r_{\mathrm{loc}}z_{t4},
z_{t1}+r_{\mathrm{loc}}z_{t5},
z_{t2}+r_{\mathrm{loc}}z_{t6}
\right),
\]
where $r_{\mathrm{loc}}=0.65$ controls the spread of the second and third pizza
options around the first.  The training design is then
\[
x_t=\sigma(\widetilde x_t/\tau_x)\in(0,1)^6,
\]
where $\tau_x=3$.  This gives a smooth
differentiable map from $\eta_\varphi({\cal I}_{t-1})$ to the design used in
the relaxed simulator. The centre-and-offset form is a policy parametrization,
not a change to the choice model. For design-trace evaluation, we
round the resulting continuous designs $x_t$ to a grid with spacing $0.01$
so that the reported attributes lie on interpretable levels in the original 
closed design space.

\section{Training and Evaluation procedure} \label{app: procedure hyper}

In this section we describe the training procedure, evaluation procedure, and
hyperparameters used for the three examples in the main manuscript.
For each example, the goal is to learn a policy network
$\eta_\varphi$, where $\varphi$ denotes the trainable network parameters.
We use the Adam optimizer with initial learning rate $lr$, and write
$\beta$ for the two Adam coefficients used to compute running averages of the
gradient and its square.  Similar to \cite{foster+imr21}, we use an
ExponentialLR scheduler with multiplicative learning-rate decay factor
$\gamma$.

For the policy network, we use the set-equivariant architecture of
\cite{foster+imr21}.  At step $t$, the network takes the previous history
$\{(x_s,y_s):s<t\}$ as input.  Each past design-outcome pair is first mapped
to an embedding of dimension $d_e$ by an encoder network with $h_e$ hidden
layers and hidden width $d_h$.  The embeddings are then pooled over the
history and passed to an emitter network with $h_m$ hidden layers and hidden
width $d_h$, which outputs the next design.  

For the source-location example, EIG and BA policies are trained from random
initialization for $N_{train}$ epochs.  At each epoch, a Monte Carlo estimate of
the relevant objective is formed using batch size $\mathcal{B}$, contrastive sample size
$L$, and outer Monte Carlo sample size $J$.  The Monte Carlo variables entering
the training objective are resampled at every epoch.

For the hyperbolic-discounting and pizza choice examples, we use a common
warm-start and fine-tune training strategy.  We first train a warm start policy net with objective \eqref{EIG-est12} for $N_{warm}=8000$ epochs.
This policy is then copied and used as the initialization for an additional fine-tune
$N_{cont}=2000$ epochs of objective-specific training.  One copy continues to
optimize the EIG lower-bound objective \eqref{EIG-est12}; this gives the
reported EIG policy.  The other copies optimize the BA objective
\eqref{eq:L-alpha-hat} for different values of $\alpha$; these give the
reported BA policies.  Thus, in these examples, EIG and BA policies are
compared after the same warm start and the same number of fine-tune epochs,
so that the difference is only the training objective in fine-tuning.  Similar 
warm-start and
fine-tuning strategies are widely used in neural-network optimization \citep{erhan2010does}. In both the warm-up and fine-tuning stages, we use windowed model checkpointing with window
length $l_w$: within each window of $l_w$ epochs, we store the model parameters
that attain the lowest training loss, and at the end of the window we reload
these parameters before continuing training.

After training, we evaluate the learned policies using batch size $\mathcal{B}_{eval}$,
contrastive sample size $L_{eval}$, and outer Monte Carlo sample size
$J_{eval}$.  These evaluations give the estimates
$\overline Z(\varphi)$, $\sqrt{S(\varphi)}$, and
\[
\widehat R(\varphi)=\frac{\overline Z(\varphi)}{\sqrt{S(\varphi)}} ,
\]
which are shown in Figures~\ref{fig loss: loc}, \ref{fig loss: hyper}, and
\ref{fig loss: pizza}.  For the design-trace figures, for each fixed parameter value we run the policy
$tr_{repeat}$ times and obtain a sample of stochastic design traces.  We treat
each trace as an ordered sequence.  If
$X^{(m)}=(x_1^{(m)},\ldots,x_T^{(m)})$ denotes the $m$th trace, we stack it as
$v^{(m)}=\operatorname{vec}(X^{(m)})$ and summarize the repeated traces by the
geometric median
\[
\widehat v
=
\arg\min_v
\sum_{m=1}^{tr_{repeat}} \|v-v^{(m)}\|_2 .
\]
The vector $\widehat v$ is reshaped to obtain the representative design trace,
and the geometric median is computed by Weiszfeld iteration
\citep{weiszfeld37} for $tr_{iter}$ iterations.  We repeat this over a prior
sample of size $tr_{sample}$, compute the DTW distance between the
representative EIG and BA traces \citep{JSSv031i07}, and display traces for
selected distance quantiles together with manually chosen parameter values.

For the pizza choice example, the geometric median is computed separately for
each of the three option trajectories, and the policy distance is the sum of
the DTW distances over the three option trajectories.

\begin{table}[!htbp]
\centering
\begin{tabular}{lc}
\toprule
Hyperparameter & Value \\
\midrule
Training epochs $N_{train}$ & $10000$ \\
Experiment steps $T$ & $30$ \\
Batch size $\mathcal{B}$ & $5000$ \\
Contrastive sample size $L$ & $100$ \\
Outer Monte Carlo sample size $J$ & $5000$ \\
Learning rate $lr$ & $5\times 10^{-5}$ \\
Adam coefficients $\beta$ & $(0.8,0.998)$ \\
Scheduler decay factor $\gamma$ & $0.98$ \\
Gradient clipping & max norm $1.0$ \\
Evaluation batch size $\mathcal{B}_{eval}$ & $5000$ \\
Evaluation contrastive size $L_{eval}$ & $100$ \\
Evaluation outer size $J_{eval}$ & $5000$ \\
Trace prior sample size $tr_{sample}$ & $500$ \\
Trace repeats $tr_{repeat}$ & $30$ \\
Geometric-median iterations $tr_{iter}$ & $100$ \\
Embedding dimension $d_e$ & $16$ \\
Hidden width $d_h$ & $128$ \\
Encoder hidden depth $h_e$ & $1$ \\
Emitter hidden depth $h_m$ & $0$ \\

\bottomrule
\end{tabular}
\caption{\small Hyperparameter settings used for the source-location example.}
\label{tab:hyper_loc}
\end{table}

\begin{table}[!htbp]
\centering
\begin{tabular}{lc}
\toprule
Hyperparameter & Value \\
\midrule
Training epochs $N_{warm}$ & $8000$ \\
Fine-tuning epochs $N_{cont}$ & $2000$ \\
Experiment steps $T$ & $20$ \\
Batch size $\mathcal{B}$ & $2000$ \\
Contrastive sample size $L$ & $1000$ \\
Outer Monte Carlo sample size $J$ & $400$ \\
Learning rate $lr$ & $5\times 10^{-6}$ \\
Adam coefficients $\beta$ & $(0.9,0.999)$ \\
Scheduler decay factor $\gamma$ & $0.98$ \\
Window length $l_w$ & $500$ \\
Evaluation batch size $\mathcal{B}_{eval}$ & $2000$ \\
Evaluation contrastive size $L_{eval}$ & $5000$ \\
Evaluation outer size $J_{eval}$ & $2000$ \\
Trace prior sample size $tr_{sample}$ & $500$ \\
Trace repeats $tr_{repeat}$ & $30$ \\
Geometric-median iterations $tr_{iter}$ & $100$ \\
Embedding dimension $d_e$ & $16$ \\
Hidden width $d_h$ & $256$ \\
Encoder hidden depth $h_e$ & $2$ \\
Emitter hidden depth $h_m$ & $2$ \\
\bottomrule
\end{tabular}
\caption{\small Hyperparameter settings used for the hyperbolic-discounting example.}
\label{tab:hyper_hyperbolic}
\end{table}

\begin{table}[!htbp]
\centering
\begin{tabular}{lc}
\toprule
Hyperparameter & Value \\
\midrule
Training epochs $N_{warm}$ & $8000$ \\
Fine-tuning epochs $N_{cont}$ & $2000$ \\
Experiment steps $T$ & $12$ \\
Batch size $\mathcal{B}$ & $500$ \\
Contrastive sample size $L$ & $500$ \\
Outer Monte Carlo sample size $J$ & $500$ \\
Learning rate $lr$ & $5\times 10^{-5}$ \\
Adam coefficients $\beta$ & $(0.8,0.998)$ \\
Scheduler decay factor $\gamma$ & $0.98$ \\
Window length $l_w$ & $1000$ \\
Evaluation batch size $\mathcal{B}_{eval}$ & $2000$ \\
Evaluation contrastive size $L_{eval}$ & $1000$ \\
Evaluation outer size $J_{eval}$ & $4000$ \\
Trace prior sample size $tr_{sample}$ & $500$ \\
Trace repeats $tr_{repeat}$ & $30$ \\
Geometric-median iterations $tr_{iter}$ & $100$ \\
Embedding dimension $d_e$ & $64$ \\
Hidden width $d_h$ & $256$ \\
Encoder hidden depth $h_e$ & $2$ \\
Emitter hidden depth $h_m$ & $2$ \\

\bottomrule
\end{tabular}
\caption{\small Hyperparameter settings used for the pizza choice example.}
\label{tab:hyper_pizza}
\end{table}

\subsection{Monte Carlo robustness check for the pizza-choice example}

\begin{sidewaysfigure}[!htbp]
    \centering
    \includegraphics[width=1.05\linewidth]
        {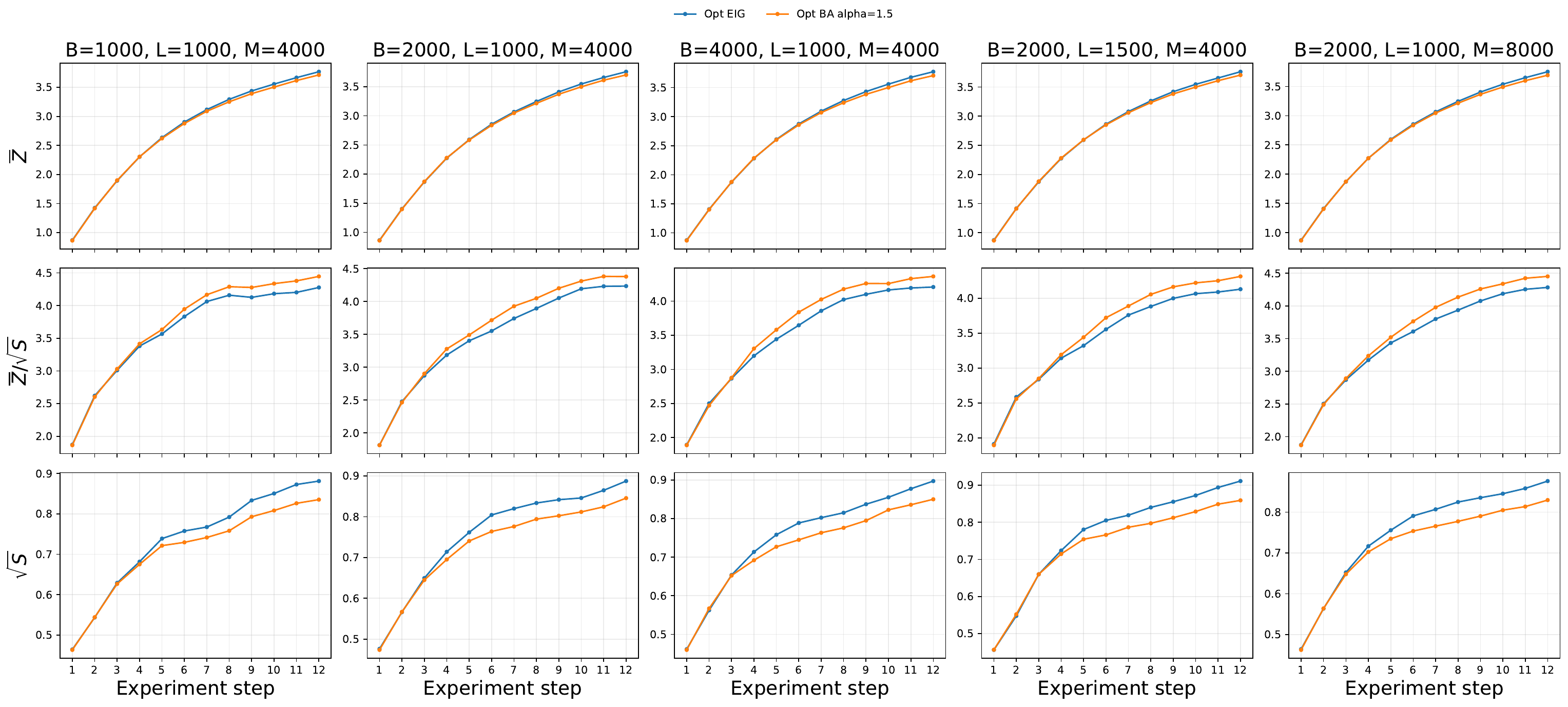}
    \caption{\small
    Monte Carlo robustness check for the pizza-choice example, comparing
    the EIG policy with the selected BA policy ($\alpha=1.5$).
    From top to bottom, the rows show $\overline Z$,
    $\overline Z/\sqrt{S}$, and $\sqrt{S}$ over the $T=12$
    experimental steps. The columns correspond to different combinations
    of $(\mathcal{B}_{\mathrm{eval}},L_{\mathrm{eval}},
    J_{\mathrm{eval}})$, as indicated above the panels.
    The qualitative comparison between the policies remains stable across
    the evaluation settings.}
    \label{fig:robustness_pizza}
\end{sidewaysfigure}

A natural concern is whether the results are sensitive to the choice of hyperparameters.  
To assess whether the comparison between the EIG and BA policies is
sensitive to Monte Carlo sample sizes, we repeated the pizza-choice objective
evaluation for several combinations of the evaluation batch size
$\mathcal{B}_{\mathrm{eval}}$, contrastive sample size $L_{\mathrm{eval}}$, and
outer Monte Carlo sample size $J_{\mathrm{eval}}$.
Figure~\ref{fig:robustness_pizza} shows that the qualitative comparison is
stable across these settings: the two policies achieve similar values of
$\overline{Z}$, while the BA policy generally has smaller $\sqrt{S}$ and
consequently larger $\overline{Z}/\sqrt{S}$.  The comparison in
Figure~\ref{fig loss: pizza} therefore does not appear to be driven by the
finite-sample Monte Carlo approximation.

\bibliographystyle{apalike}
\addcontentsline{toc}{section}{\refname}
\bibliography{Deep-adaptive-design-bias}

\end{document}